\documentclass[10pt,twocolumn,twoside]{IEEEtran}
\usepackage[linesnumbered,ruled, vlined]{algorithm2e}
\usepackage{amsmath}
\usepackage{amsthm}
\usepackage{enumitem}
\usepackage[margin=1.0in]{geometry}
\usepackage{makecell}
\usepackage[table]{xcolor} 
\usepackage{tikz}
\usetikzlibrary{fit,calc}
\newcommand*{\tikzmk}[1]{\tikz[remember picture,overlay,] \node (#1) {};\ignorespaces}
\newcommand{\boxit}[1]{\tikz[remember picture,overlay]{\node[yshift=3pt,fill=#1,opacity=.25,fit={(A)($(B)+(.87\linewidth,0\baselineskip)$)}] {};}\ignorespaces}
\colorlet{pink}{red!40}
\colorlet{blue}{cyan!60}

\newtheorem{theorem}{Theorem}
\newtheorem{lemma}{Lemma}
\newtheorem{corollary}{Corollary}
\newtheorem{definition}{Definition}
\newtheorem{result}{Result}

\newcommand{\suppMat}{supporting online material}
\newcommand{\secProof}{S1}
\newcommand{\secEfficiency}{S2}
\newcommand{\secTimeCost}{S2.1}
\newcommand{\corollaryOne}{1}
\newcommand{\theoremOne}{1}
\newcommand{\theoremTwo}{2}
\newcommand{\theoremThree}{3}
\newcommand{\tableWorstCase}{S1}

\newcommand{\Algorithm}{1}
\newcommand{\LoopInit}{9}
\newcommand{\Init}{1}
\newcommand{\LeafIdentificationCondition}{19}
\newcommand{\StartTransitionCondition}{26}
\newcommand{\Figure}{2}

\title{A memory and communication efficient algorithm for decentralized
  counting of nodes in networks}

\author{Arindam Saha, James A. R. Marshall and Andreagiovanni Reina%
  \thanks{A. Saha, J.A.R. Marshall, and A. Reina are with the Department
of Computer Science, University of Sheffield, S1 4DP, UK;
e-mails: (a.saha@sheffield.ac.uk, james.marshall@sheffield.ac.uk, a.reina@sheffield.ac.uk).}%
\thanks{Manuscript received ; revised.}}

\begin{document}

\maketitle

\begin{abstract}

Node counting on a graph is subject to some fundamental theoretical
limitations, yet a solution to such problems is necessary in many
applications of graph theory to real-world systems, such as collective
robotics and distributed sensor networks. Thus several stochastic  and
na{\"i}ve deterministic algorithms for distributed graph size
estimation or calculation have been provided. Here we present a
deterministic and distributed algorithm that allows every node of a
connected graph to determine the graph size in finite time, if an
upper bound on the graph size is provided. The algorithm consists in
the iterative aggregation of information in local hubs which then
broadcast it throughout the whole graph. The proposed node-counting
algorithm is on average more efficient in terms of node memory and
communication cost than its previous deterministic counterpart for
node counting, and appears comparable or more efficient in terms of
average-case time complexity. As well as node counting, the algorithm
is more broadly applicable to problems such as summation over graphs,
quorum sensing, and spontaneous hierarchy creation.
\end{abstract}


\section{Introduction}

\IEEEPARstart{A}{ll} decentralized systems share the common aspect of being comprised of a network of units (which can be considered as graph nodes) that rely on local and partial information which they can gather from the subset of devices in their communication range (communication links can be represented as graph edges). An open challenge is to allow the units of these large-scale decentralized systems to estimate properties of the entire group.

A fundamental property that is crucial for the design and the efficient functioning of several systems is the system size, that is, the number of units in the system. Computing the exact network size in finite time with a decentralized algorithm with finite complexity is proved to be impossible~\cite{Hendrickx2011}. Previously proposed solutions are therefore stochastic algorithms that only give an approximation of the system size, providing the possible advantages of robustness and speed. Deterministic algorithms provide the exact solution in a finite time, however, they may rely on stringent assumptions on the communication network topology. An overview of the existing algorithms is provided in Section~\ref{sec:stateart}.
%
We propose, in Section~\ref{sec:algo}, a new decentralized deterministic algorithm, the \textit{aggregate-and-broadcast} (AnB) algorithm, that iteratively aggregates the node counts into a small number of local hubs which finally broadcast the count throughout the whole network.
The AnB algorithm allows the nodes to compute the exact network size in a finite time when an upper bound is provided.
In other words, the network size computed by the AnB algorithm is exact up to a limit that is bounded by the algorithm's execution time, as proved in the \suppMat{}.
%
The algorithm relies on the only two assumptions of a connected network and uniquely identifiable units (i.e. unique ID), and requires minimal computation and communication capabilities of the units. The algorithm performance is analysed and when possible compared with previous algorithms in terms of time, communication, and memory costs (see Section~\ref{sec:analysis}). The results indicate that the AnB algorithm is scalable, efficient, and accurate, with better performance than the existing algorithms in terms of smaller memory and communication costs. Therefore, as discussed in Section~\ref{sec:conclusion}, the AnB algorithm can be beneficial for systems with constrained memory and communication, and has the potential to be employed in numerous application cases and impact a large variety of decentralized systems.



\section{Problem statement}
Consider a connected network $\mathcal{G}=(\mathcal{V}, E)$, where
$\mathcal{V}=\{1,\dots,N\}$ is the set of nodes in the network and
$E \subseteq \mathcal{V}\times \mathcal{V}$ is the set of the edges of
the network. The edges describe undirected and unweighted
communication links between nodes, i.e. $(u,v) \iff (v,u) \in E$. Each
node can only communicate at synchronous timesteps with its neighbors,
where the set of neighbors of the generic node $v$ is defined as
$\mathcal{N}_i =\{ u \in \mathcal{V} | (v,u) \in E\}$. We
  assume $\mathcal{G}$ to be time-invariant.
Each node is characterized by a unique identifier (id). Each node
knows an upper bound $N_{\text{max}}$ of the network size, such that
$N_{\text{max}} \ge N$. In this paper, we propose an algorithm to be
executed by every node of the network to allow them to compute the
network size $N$ in a finite amount of time $t_{\text{max}} \le
4N_{\text{max}} +1$.
Note that knowledge about $N_{\text{max}}$ is only necessary in order
to bound the execution time of the algorithm to $t_{\text{max}}$. This
is required due to the results reported by Hendrickx et
al.~\cite{Hendrickx2011} who have proved that it would be otherwise
impossible for a finite complexity algorithm to correctly count the
number of nodes (see discussion in Sec.~\ref{subsec:stopping}).




\section{State of the art}\label{sec:stateart}

Most of the algorithms proposed to estimate the size of the network rely on stochastic methods. The most common approach relies on executing variations of random walks on the network~\cite{Ganesh2007, Gjoka2010, Katzir2014, Musco2016a}. In particular, Ganesh et al.~\cite{Ganesh2007} used continuous time random walks to obtain a target number of redundant node samples. The time required to obtain such a sample was then used to estimate the network size. In a different study, Gjoka et al.~\cite{Gjoka2010} compared various weighted random walk techniques. The study identified efficient methods to identify various macroscopic properties of the network by simulating weighted random walks on the network (e.g. Metropolis-Hastings Random Walk and Re-Weighted Random Walk). Similarly, Katzir et al.~\cite{Katzir2014} proposed a method based on simulating multiple simultaneous random walks in order to estimate the size of the network. Building upon this work, Musco et al.~\cite{Musco2016a} proposed an algorithm where multiple nodes execute random walks and compute the network size based on the degrees of the nodes encountered. Notable stochastic algorithms which do not involve random walks rely on either average consensus~\cite{Jelasity2004} or on order statistics consensus~\cite{Lucchese2015, Varagnolo2014f, Lucchese2015a}.

One of the shortcomings of stochastic algorithms is that their run-times depend on the desired accuracy of the results. Therefore, for applications where the size of the network is required to a high degree of accuracy, stochastic algorithms might take a long time to converge. For instance, the number of dynamical attractors in Boolean networks and their periodicities depend on whether the network size is even or odd, prime or composite~\cite{Drossel2005}. Since dynamics on such networks are crucial in studying social networks, neural networks and gene and protein interaction networks~\cite{Green2007, Cheng2009, Davidich2008, Kauffman2004, Kauffman1969}, accurate knowledge of the network size is crucial. In such scenarios, deterministic algorithms to estimate the network size are better suited.

To the best of our knowledge, the number of deterministic algorithms for decentralized network node counting is very limited. One of the most trivial algorithms is the \textit{All-2-All} method, as alluded to in Ref.~\cite{Evers2011}. It consists in having each node broadcasting a unique id together with all ids that it has already received so far. This simple algorithm is the most efficient algorithm we are aware of for deterministic network node counting on general network topologies. Other algorithms for node counting have been proposed for networks with specific topologies. For example, an algorithm inspired by the Breadth-First-Search (BFS) algorithm can be used on a tree network. In 2003, Bawa et al.~\cite{Bawa2003} generalized such an algorithm so that it could be implemented on a network with a general topology. In their paper, the authors propose three different algorithms which may be used for computing various aggregates across the network. While the proposed algorithms are efficient, they investigated a different problem. They focus on the situations when the network size or the other aggregate quantities are sought by a single node of the network. When every node requires the size information, repeating the algorithm of~\cite{Bawa2003} on every node becomes less efficient than the All-2-All method, as described in Sec.~\ref{sec:analysis}.

\section{The aggregate-and-broadcast algorithm}\label{sec:algo}

We propose the aggregate-and-broadcast (AnB) algorithm, a deterministic algorithm for the simultaneous and decentralized determination of the size $N$ of a finite connected network by all its nodes. We assume that each node of the network has a unique id, can communicate only with its immediate neighbors, and knows $N_{\text{max}}$, the upper bound of the network size. Other than that, we make no prior assumptions about the topology of the network nor prior knowledge of the node. The underlying idea of the AnB algorithm is inspired by the standard node-counting method on a tree by its root. In a tree, the counts of the leaves are assimilated by their respective parents and then the leaves are iteratively pruned. Applying such an algorithm on a graph with a general topology poses a challenge since a strict hierarchy does not exist among the nodes. To overcome this problem, we add a step in each iteration where, based on the degree of its neighbors, each node determines its local hierarchy which, in turn, determines whether it should be pruned or not.

In the next subsections, we describe the proposed AnB algorithm in detail. We start with an overview of the entire algorithm in the next subsection. In subsections \ref{subsec:pre-iteration} and \ref{subsec:iteration}, we describe the pre-iteration steps (which include variable initialization) and the iteration steps of the algorithm respectively. Finally, in subsection \ref{subsec:remarks} we compare the AnB algorithm to the standard node counting algorithm in trees and make some further remarks about the proposed algorithm. The correctness of the AnB algorithm is proved in Sec.~\secProof{} of the \suppMat{}.

\subsection{An overview of the AnB algorithm}
\label{subsec:Overview}

Prior to the iterative steps, the nodes of the network are initialized as follows. The behavior of a node with id~$i$ at any particular instant is determined by its state~$s_i$ which can take one of four values during the course of the algorithm: `active'~$(A)$, `leaf'~$(L)$, `residue'~$(R)$, or `inactive'~$(I)$. The state of each node is initialized to $s_i = A$. Each node also starts with a local node counter $c_i=1$. This variable keeps track of the primary number of nodes in the network as locally known to the node at any point in time. Since, at the beginning of the algorithm, each node is aware only of its own existence, the counter is initialized to $1$. As the algorithm progresses, the node gathers information about the changing state of nodes (equivalent to the nodes getting `pruned') from its neighbors and updates the value in $c_i$. Additionally, each node also has the following other internal variables: the set of its neighbors $\mathcal{N}_i$, its effective neighborhood $\mathcal{E}_i$, effective degree $ e_i $, the set of residues ${\mathcal{R}}_i$ and final node count $n_i$. Among these, the first three variables are initialized to be empty sets $\mathcal{N}_i = \mathcal{E}_i = \mathcal{R}_i = \emptyset$, and the effective degree and final count variable are initialized as $e_i = n_i = 0$. We describe these variables in further details in the following paragraphs.

From the perspective of a node, the AnB algorithm is divided into two phases: `pre-reduction' and `post-reduction'. A node is said to be in pre-reduction phase when its state is either $s_i = A$ or $s_i = L$. As this phase progresses, a node in `active' state updates its local counter $c_i$ by locally accumulating information from `leaf' neighbors getting `pruned' until the node itself changes its state to $s_i = L$ and becomes a `leaf' node. Note that, here the term `leaf' is used to denote a node which is about to be `pruned' from the network; and not necessarily a node with only one neighbor. In the next iteration, each leaf node, depending on their effective neighborhood $\mathcal{E}_i$, again changes its state to either (a)~$s_i = I$ and gets `pruned', or (b)~$s_i = R$ and becomes a residue node.

At the end of pre-reduction phase, the nodes of the network are either in residue $(s_i = R)$ or inactive $(s_i = I)$ states. These states can be considered analogous to the `root' and the `pruned leaves' of a tree network respectively. The residue nodes contain parts of the total count of nodes in the network. This is similar to the root of a tree network which contains the total node count of the entire tree after all the nodes have been pruned. This information is then broadcast across all other nodes and assimilated to give the final node count of the network. To do this, each residue node constructs a `broadcast message' $b_i$, sends it to all its neighbors and changes its state to $s_i = I$. This broadcast message is then relayed by all nodes---irrespective of their state $s_i$---across the network. A node that receives a broadcast message adds the partial count to its final count variable $ n_i $, and keeps track of the residue nodes to avoid double counting. Thus, after sufficient time $t_{\text{max}}$, the variable $n_i$ gives the total count of all nodes in the network. The exact rules for updating the local count variables $c_i$ and state $s_i$, and constructing the broadcast message $b_i$ are given in subsection \ref{subsec:iteration}, and the details about the stopping criteria are provided in subsection~\ref{subsec:stopping}

\begin{figure*}
	\includegraphics[width=\textwidth]{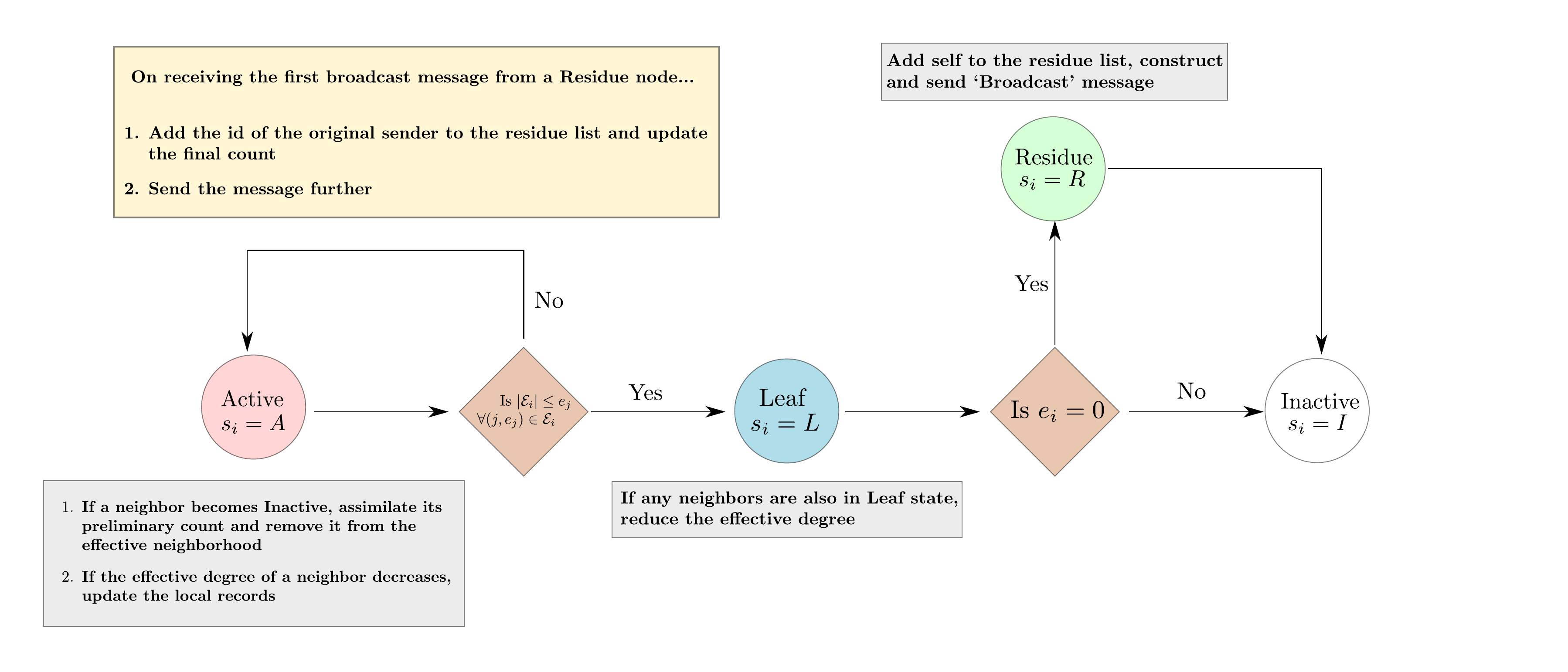}
	\caption{Schematic flowchart depicting the finite state machine of each node of the network executing the AnB algorithm. Note that the colors of the circles correspond to the colors of the section in Algorithm~\ref{al:Algorithm}. Also, the steps outlined in the yellow box are carried out by all nodes irrespective of their state.}
	\label{fig:flowchart}
\end{figure*}

\begin{algorithm}
	\label{al:Algorithm}
	\caption{The aggregate-and-broadcast (AnB) algorithm for network node counting}
        \small
	\BlankLine
	\BlankLine
	let $s_i \gets A$, $c_i \gets 1$, $\mathcal{N}_i \gets \emptyset$, $\mathcal{E}_i \gets \emptyset$, $\mathcal{R}_i \gets \emptyset$, $n_i \gets 0$\;  \label{st:init}
	send message $m_{i,\text{\upshape echo}}$\;
	\ForEach{message $ m_{j,\text{\upshape echo}} $ received}{
		$ \mathcal{N}_i \gets \mathcal{N}_i \cup \{ j \} $\;
		\label{st:neighborhood_identification} 
	}
	set $ e_i \gets |\mathcal{N}_i| $\;
	send message $m_{i,\text{\upshape degree}} = e_i$\;
	\ForEach{message $ m_{j,\text{\upshape degree}} $ received}{
		$ \mathcal{E}_i \gets \mathcal{E}_i \cup \{ \left( j, e_j \right) \} $\;
		\label{st:initial_effective_neighborhood} 
	}
	
	\BlankLine
	\For{$t = 0;\ t < t_{\text{max}};\ t = t + 1$}{ \label{st:loop_init}
		\BlankLine
		\tikzmk{A}
		\uIf(\tcp*[f]{Executed if node is in Active state}){$s_i = A$}{
			\ForEach{message $m_{j,h}$ received}{
				\If{$h = \text{\upshape count}$}{
					let $c_i \gets c_i + m_{j,\text{\upshape count}}$\; \label{st:c_i_update}
					let $\mathcal{E}_i \gets \mathcal{E}_i \setminus \left\{(j, e_j) \right\}$\;
					let $ e_i \gets e_i - 1 $\;
					send message $m_{i,\text{\upshape reduce}}$\;
				}
				\If{$h = \text{\upshape reduce}$}{
					let $\mathcal{K}_j \gets \left( j, e_j - 1 \right)$, where $\mathcal{K}_j \in \mathcal{E}_i$\; \label{st:E_i_update}
				}
			}
			\If{messages $ m_{j,\text{\upshape count}} $ are not received \upshape{\textbf{and}} $e_i \leq e_{j} ~ \forall (j, e_j) \in \mathcal{E}_{i}$\label{st:leaf_identification_condition}}{
				send message $m_{i,\text{\upshape leaf}}$\;
				let $s_{i} \gets L$\;
			}
		}
		\tikzmk{B}
		\boxit{pink}
		\BlankLine
		\tikzmk{A}
		\uElseIf(\tcp*[f]{Executed if node is in Leaf state}){$s_i = L$}{
			\ForEach{message $m_{j,\text{\upshape leaf}}$ received}{
				let $e_i \gets e_i - 1$\;
			}
			\eIf{$e_i = 0$\label{st:start_transition_condition}}{
				let $s_i \gets R$\;
			}{
				send message $m_{i,\text{\upshape count}} = \frac{c_i}{e_i}$\;
				let $s_i \gets I$\; \label{st:end_transition_condition}
			}
		}
		\tikzmk{B}
		\boxit{cyan}
		\BlankLine
		\tikzmk{A}
		\ElseIf(\tcp*[f]{Executed if node is in Residue state}){$s_i = R$}{
			let $\mathcal{R}_i \gets \mathcal{R}_i \cup \left\{ i \right\}$ and $n_i \gets n_i + c_i$\;
			send message $m_{i,\text{\upshape broadcast}} = \left( i, c_i \right)$ \;
			let $s_i \gets I$\;
		}
		\tikzmk{B}
		\boxit{green}
		\BlankLine
		\tikzmk{A}
		\ForEach( \tcp*[f]{Executed at each iteration}){message $m_{j,\text{\upshape broadcast}} = (k, c_k)$ received}{
			\If{$(k \notin \mathcal{R}_i)$ \label{st:residue_check}}{
				let $\mathcal{R}_i \gets \mathcal{R}_i \cup \left\{ k \right\}$ and $n_i \gets n_i + c_k$\;
				send message $m_{i,\text{\upshape broadcast}} = m_{j,\text{\upshape broadcast}}$ \label{st:rebroadcast}
			}
		}
		\tikzmk{B}
		\boxit{yellow}
	}
\end{algorithm}

\subsection{Pre-iteration steps}
\label{subsec:pre-iteration}

We now describe the AnB algorithm (see Algorithm~\ref{al:Algorithm}) in detail. The actions taken by a node $i$ in a particular step are determined by its internal variables and the messages it receives from its neighbors, \textit{i.e.} the nodes in $\mathcal{N}_i$. 

Any message sent by a node is denoted as $m_{i,h}$, where $i$ is the sender of the message and $h$ is the `type' of the message. The `type' of the message determines the action to be taken by the receiver of the message. The various types of messages and their roles are summarized in Table~\ref{tab:message_types}. Note that every message is broadcast to the entire neighborhood $ \mathcal{N}_i $ and thus, can be accessed by all nodes in $\mathcal{N}_i$.

After the initialization of all internal variables, each node of the network identifies its neighborhood. To do so, it sends a message $m_{i,\text{echo}}$ indicating its presence to all its neighbors. It then receives similar messages $ m_{j,\text{echo}} $ from other nodes. The set of all nodes from which such a message is received is then identified as the neighborhood $ \mathcal{N}_i $ (Line~\ref{st:neighborhood_identification}).

One of the most crucial internal variables for the node is its effective degree $e_i$ which is the number of its neighbors which are in the active state ($s_i = A$). Since all nodes start in the active state, the initial effective degree of the node is the number of elements in its neighborhood: $e_i = \left| \mathcal{N}_i \right|$. In addition to its own effective degree, the node also needs to be aware of the effective degrees of those neighbors which are in active state. The node keeps track of this information in form of its effective neighborhood,
\begin{equation}
	\mathcal{E}_i = \left\{ (j,e_j) : j \in {\mathcal{N}}_i \text{ and } s_j = A \right\}.
\end{equation}
Therefore, $\mathcal{E}_i$ is a set of tuples where the first element of the tuple is the id of an active neighbor of $i$ and the second element is the effective degree of the neighbor.

\begin{table*}
	\caption{The different types of messages $m_{i,h}$ used in the AnB algorithm}
	\label{tab:message_types}
	\centering
	\begin{tabular}{|c|c|c|}
		\hline
		\rowcolor{lightgray}
		$h$ & Content & Role of the message \\
		\hline
		echo & - & Indicates the presence of the sender $ i $ \\
		degree & $e_i$ & Sends the initial effective degree $ e_i $ of the sender $ i $ \\
		leaf & - & Indicates the transition of the sender $ i $ to leaf state \\
		count & $c_i$ & Sends the local count $ c_i $ of the sender $ i $ \\
		reduce & - & Indicates the reduction of effective degree $ e_i $ of the sender $ i $ \\
		broadcast & $(k, c_k)$ & Sends or relays the broadcast message \\
		\hline
	\end{tabular}
\end{table*}


 The identification of neighborhood also allows the node to compute its initial effective degree $ e_i = \left| \mathcal{N}_i \right| $ and to send it to its neighbors as $ m_{i,\text{degree}} $. Thereafter, a node $ i $ receiving a message $ m_{j,\text{degree}} $ updates its effective neighborhood $ \mathcal{E}_i $ as described in Line~\ref{st:initial_effective_neighborhood}.

\subsection{Iteration steps}
\label{subsec:iteration}

After the pre-iteration steps of Sec.~\ref{subsec:pre-iteration}, the node $ i $ enters an iterative phase where its steps are determined by its state $ s_i $. The details of these state-dependent steps are illustrated in the finite state machine of Figure~\ref{fig:flowchart} and are elaborated as follows.


\begin{itemize}
	
	\item  \textbf{Active nodes:} Each active node $ i $ with $s_i = A$ first detects any change in its neighborhood. This change can be of two types: (a) Either some of its neighbors are transitioning to inactive state, which is indicated by a message of type $ h = \text{count} $; or (b) the effective degree of some of its neighbors is being reduced, which is indicated by a message of type $ h = \text{reduce} $. Therefore, upon receipt of a message $ m_{j,\text{count}} $ (of type $ h = \text{count} $), the node $ i $ excludes the sender from its effective neighborhood $ \mathcal{E}_i $, decreases its effective degree $ e_{i} $ by 1 and assimilates the contents of the message in its local count (Line~\ref{st:c_i_update}),
	  \begin{equation}
	    c_i = c_i + m_{j,\text{count}}.
	    \label{eq:c_i_update}
	  \end{equation}
	Since the effective degree of node $ i $ is decreased by 1, it sends a message $ m_{i,\text{reduce}} $ to its neighbors. For each message of type $ h = 5 $ received, the node updates the record of the effective degree corresponding to the sender of the message (Line~\ref{st:E_i_update}).
	
	After processing the incoming messages, the node $ i $ checks that two clauses are true. It checks whether it has not received this time step any message of type $h=\text{count}$, which would indicate that its measure $\mathcal{E}_i$ may be temporarily incorrect, and checks weather its own effective degree $ e_i = |\mathcal{E}_i| $ is the minimum among all its neighbors which are in active state (Line~\ref{st:leaf_identification_condition}). If both conditions are met, the node sends a message $ m_{i,\text{leaf}} $ and changes its state to $ s_i = L $; otherwise, the node stays in the active state for the next iteration.
		

  \item \textbf{Leaf nodes:} The node $ i $ in state $ s_i = L $ stays in this state for exactly one iteration and then changes its state to either $s_i = R$ or $s_i = I$. First, it processes any incoming message of the type $ h = \text{leaf} $. The reception of any such message implies that some of its neighbors have transitioned to the leaf state in the same time step, and are therefore no longer in the active state. For each message $ m_{j,\text{leaf}} $ received, the effective degree $ e_i $ of the node is reduced by one. After processing all incoming messages, the node $ i $ changes its state; if the effective degree $ e_i = 0 $, it change its state to $ s_i = R $ otherwise, it sends the message
  \begin{equation}
    m_{i,\text{count}} = \frac{c_i}{e_i}
    \label{eq:reduction_message_construction}
  \end{equation}
  and changes state to $s_i = I$ (Lines~\ref{st:start_transition_condition}--\ref{st:end_transition_condition}).
  
    \item \textbf{Residue nodes:} Each node $ i $ in state $s_i = R$ updates its residue set $ \mathcal{R}_i $ with its own id $ i $ and the total node counter $ n_i $ adding its local counter $ c_i $. It then broadcasts a message $m_{i,\text{broadcast}} = \left( i, c_i \right)$ and changes its state to $s_i = I$.
    
    \item \textbf{All nodes:} While the previous steps are executed by nodes in a specific state, the following steps are executed by all nodes of the network at each iteration irrespective of their state. Whenever a node $ i $ receives a message $ m_{j,\text{broadcast}} = \left( k, c_k \right)$ from any of its neighbors, it checks if node $ k $ is in the residue set $ \mathcal{R}_{i} $. If $ k \notin \mathcal{R}_i $, it means that the recipient node has received the message $ \left( k, c_k \right) $ earlier. In this case, the node $ i $ adds $ k $ to its residue set $ \mathcal{R}_i = \mathcal{R}_i \cup \{ k \} $, adds the corresponding local count $ c_k $ to its final node count $ n_i = n_i + c_k $ and finally relays the message forward by sending message $ m_{i,\text{broadcast}} = m_{j,\text{broadcast}} $.
    
    

\end{itemize}

After a sufficient number of time steps $t_{\text{max}}$, all nodes converge to the same final count $ n_i $ equal to the network size $ N $. A detailed analysis of the convergence time is provided in Sec.~\secTimeCost{} of the \suppMat{}.

\subsection{Stopping criteria}
\label{subsec:stopping}

The AnB algorithm terminates when sufficient time, $t_{\text{max}}$, has passed. This $t_{\text{max}}$ should be sufficiently large so that each broadcast message reaches every node of the network. However, determining an exact value for $t_{\text{max}}$ is impossible as reported by Hendrickx et al.~\cite{Hendrickx2011} who have shown that it is impossible for a finite complexity algorithm to correctly estimate the size of a network with probability one. If $t_{\text{max}}$ could be exactly determined for the network, we would be absolutely sure that each residue message has reached every node and hence, each node is aware of the size of the network. This would be in direct violation of the aforementioned result. However, depending on the prior knowledge about the network, various estimates of $t_{\text{max}}$ can be made as follows. In Sec.~\secProof{} (Corollary~\corollaryOne{}) of the \suppMat{}, we show that the maximum time required for all nodes to reach the final state, i.e., the inactive state, has the above boundary of $t_r = 3N+2$. It is also trivial that the number of time steps required to broadcast a message across a network of size $ N $ is, in the worst-case, $t_b = N-1$. Therefore, $t_{\text{max}}$ is bounded above by $ t_r + t_b = 4N+1 $. Hence, if an overestimate $ N_{\text{max}} $ of the network size is known apriori, we can set $ t_{\text{max}} = 4N_{\text{max}}+1 $ to know the exact size of the network in finite time.

\subsection{Remarks on the AnB algorithm}
\label{subsec:remarks}

As shown in Figure~\ref{fig:flowchart}, a node spends exactly one iterative step as a leaf, and at most one iterative step as a residue node. Therefore, a typical node spends most of its iterative steps in either active or inactive states.

We can now elaborate on the similarities between the proposed AnB algorithm and the standard node-counting method on a tree network which were indicated earlier. Nodes in a tree network can also be classified into four categories analogous to those in the AnB algorithm: (a)~the root (similar to $ s_i = R $), (b)~leaves (similar to $ s_i = L $), (c)~pruned leaves (similar to $ s_i = I $) and (d)~other nodes still in the network (similar to $ s_i = A $). In a tree network, leaves are easily identified as nodes with degree one. Since this is not true for a general network, we use the condition in Line~\ref{st:leaf_identification_condition} to identify, at each iteration step, the nodes which are to be labeled as leaves. After a node has been identified as a leaf in a tree network, it passes on its local count to its parent and gets transformed to a pruned leaf. In a tree network, the parent of each node is unique. However, in a general network, a leaf node may have more than one parent. Therefore, in the AnB algorithm, the local count of each leaf is divided equally among all parents to avoid over-counting number of nodes. Once the counts have been passed on, the leaf node becomes an inactive node, similar to the pruned leaves in a tree network. If there are no active neighbors (`parents') to which a node can pass on its local count, it becomes a residue node, which is similar to the root of the tree. While the structure of the tree implies that there can be only one root of a tree, there is no such restriction for a general network. Hence, the count of the size of a general network gets concentrated into the residue nodes which is then broadcast and recombined in the final stages of the AnB algorithm.

It is to be noted that each node checks for the reception of a message of type $ h = \text{broadcast} $ at each iteration. This is necessary because messages of type $ h = \text{broadcast} $ carry the node count of a part of the network as counted by a residue node. Therefore, all nodes which receive such a message should add it to their final count and send it further. This is in contrast with the other types of messages which are intended only for nodes in active or (as in case of $ h = \text{leaf} $) leaf states.


\section{Analysis of the algorithm}\label{sec:analysis}

In this section, we demonstrate the correctness of the AnB algorithm and analyze the algorithm performance in terms of time, communication, and memory costs against the known node-counting algorithms. We do not compare AnB with stochastic algorithms which only compute an estimate of the network size that increases over time, but we limit our comparison against algorithms that return the exact node count in a finite time: the All-2-All algorithm and the Single Tree (ST) algorithm~\cite{Bawa2003}.

The All-2-All algorithm is, to the best of our knowledge, the only known deterministic algorithm for node counting which can work on any type of connected network regardless of its topology. In the All-2-All algorithm, each node broadcasts its id, and all received ids, to all its neighbors and every node counts the number of received unique ids.

The ST algorithm, instead, is the most efficient of the three algorithms proposed by Bawa et al. in~\cite{Bawa2003}. Despite being stochastic, the ST algorithm is proved to return the exact network size in a finite time. The ST algorithm, similarly to AnB, relies on the construction of a tree-like hierarchy. However, in its original form, the ST algorithm allows only a single node to compute the network size. In order to allow all the nodes of the network to know the network size, the ST algorithm can be extended in the following two ways: (a) one randomly selected node executes the ST algorithm and then broadcasts the computed size to all other nodes; or (b) all the nodes of the network simultaneously execute the ST algorithm and compute the network size independently. Employing alternative (a) requires the nodes to be able to select in a decentralized way which node will execute the ST algorithm. Decentralized node-selection adds a new problem which may require further assumptions on the network topology or on the initial knowledge of the nodes~\cite{patterson2010leader}. Therefore, in our comparison against the ST algorithm, we employ alternative (b) by which every node makes an independent count of the network size.

We provide a comparison both as worst-case algorithm complexity and with generic analytical equations for each type of cost. When such analytical solutions are not possible, we provide the results of numerical simulations for specific graph topologies. In fact, the AnB algorithm is proved to work on any connected graph regardless on the graph topology. Through our analysis, we highlight the differences in performance for each topology.

\subsection{Correctness of the AnB algorithm}

In Sec.~\secProof{} of the \suppMat{}, a detailed proof of correctness of the algorithm is provided. A brief sketch of the proof is as follows. We begin by identifying a sequence of time steps of the algorithm when the variables $ e_i $ and $ \mathcal{E}_i $ correctly give correct information about the neighborhood of the node $ i $ (see Theorem~\theoremOne{}). We say that, at these time steps, the network is in the \textit{resting state}. We then show that, as the network progresses from one resting state to another, the number of active states decreases. During this process, the information about their local node counts $ c_i $ gets concentrated into the nodes which pass through the residue state (see Theorem~\theoremTwo{}). Therefore, when no active nodes are present in the the network, the information about the size of the network is concentrated in the nodes which passed through the residue state. This information is then broadcast throughout the network and is accumulated by each node (see Theorem~\theoremThree{}).

\subsection{Comparison with other algorithms in terms of complexity}


We compare the efficiency of the AnB algorithm against the All-2-All and the Single Tree (ST,~\cite{Bawa2003}) algorithms in terms of three aspects: (a)~the time required to compute the network size by every node, (b)~the number of messages sent by all nodes (i.e. the communication cost), and (c)~the minimum amount of memory required by each node to execute the algorithm (i.e. the memory cost).

Note that, it is difficult to compare the efficiency of AnB against most other stochastic algorithms because their efficiency depends on the desired accuracy of the results. The more accurate we want the results to be, the longer the stochastic algorithms should run, at the cost of increased time and/or communication costs. On the other hand, deterministic algorithms like ours give accurate results in a finite time and make possible asymptotic performance analysis.

%

The efficiency results for the AnB algorithm are derived in Sec.~\secEfficiency{} of the \suppMat{} and reported in Table~\ref{tab:generalcase}. We derive exact results for the communication and memory costs. Instead, computing a precise equation of the time cost is difficult, as it depends strongly on the topology of the network which evolves at every time step (see discussion in Sec.~\secTimeCost{}).  Through Theorem~\theoremThree{} in Sec.~\secProof{}, we computed the upper bound of the time complexity of AnB. To analyse the exact performance in terms of time, instead, we computed a set of numerical simulations on various graph topologies whose results are shown in Figure~\ref{fig:Scaling}.
In particular, we implemented and tested the AnB algorithm on four different types of random networks as listed in Table~\ref{tab:network_types}.
The results of our analysis show a qualitative difference in algorithm performance as a function of the network topology.
We employed these numerical simulations to compare the temporal performance of AnB with the All-2-All algorithm and to make general considerations on the execution time of the AnB algorithm (see also Sec.~\secTimeCost{}).

\begin{table*}
  \centering
  \renewcommand{\arraystretch}{1.3}
  \caption{Exact costs for the two algorithms for a general network with diameter $D$, average degree $d$, and $r$ residue nodes. For memory cost, we indicate the individual degree $d_i$ for the generic node $i$. The AnB algorithm is more efficient than the All-2-All and the ST methods in terms of memory and communication. Analytical solution for time is out of reach and we provide numerical results in Figure~\ref{fig:Scaling}.}\vspace{0.3cm}
\begin{tabular}{ |l|c|c|c|  }
  \hline
  \rowcolor{lightgray}
  Algorithm & Time & Communication & Memory \\
  \hline
  AnB & numerically in Fig.\,\ref{fig:Scaling}& $N(4+r+d)-r$ & $(2d_i+r+5) \log(N)$\\
  All-2-All & $D$ & $N^2$ & $N \log(N)$\\
  ST & $2D$ & $2N^2$  & $2N \log(N) + d_i N$ \\
  \hline
\end{tabular}
\label{tab:generalcase}
\end{table*}

\begin{figure}[ht]
  \centering
  \includegraphics[width=0.35\textwidth]{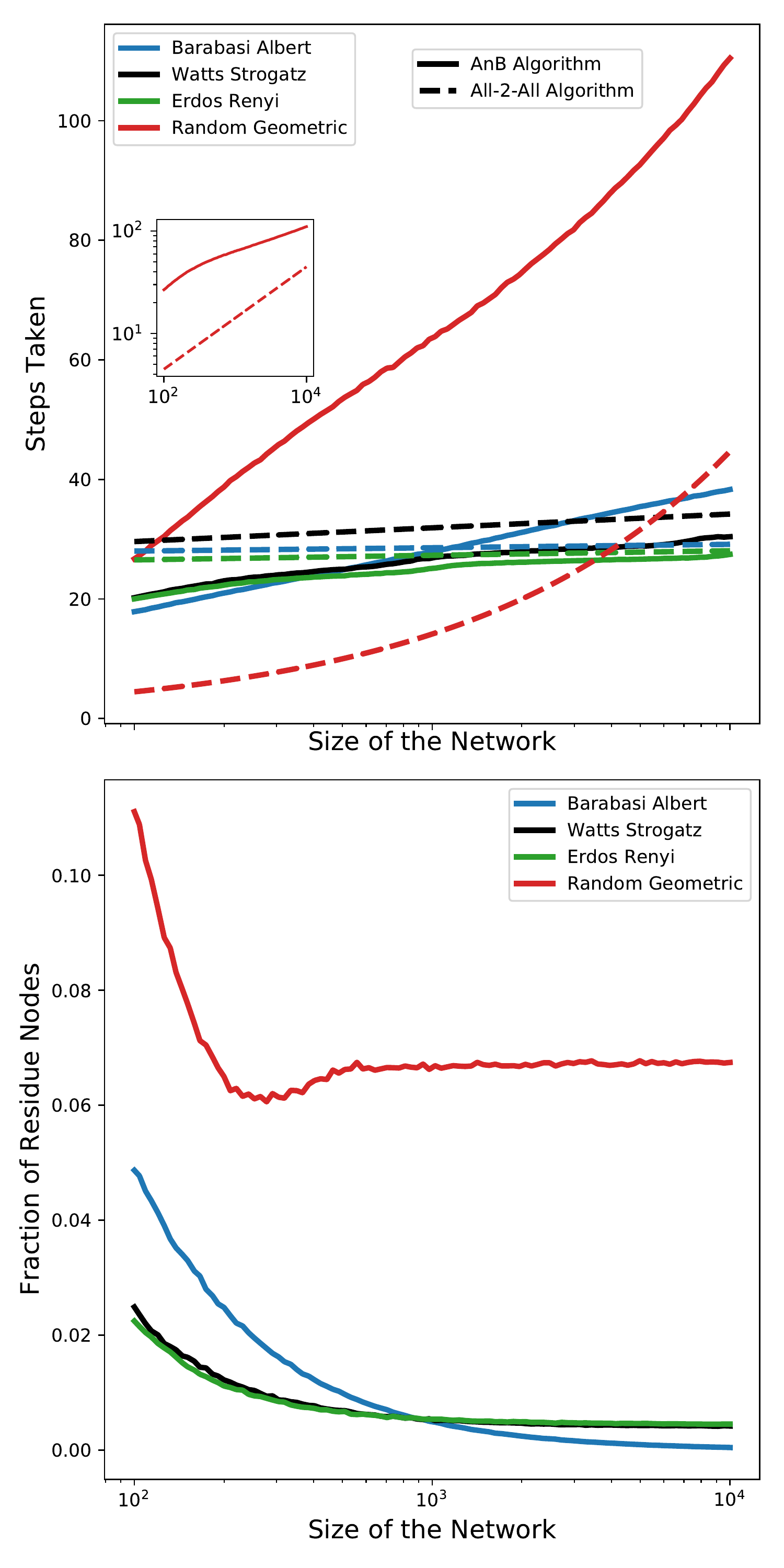}
  \caption{Numerically estimated time costs of the AnB algorithm. The left panel shows, on a log-linear scale, the total number of iterative steps taken by the AnB algorithm for different random networks (solid lines). The dashed lines show the scaling of the time for the All-2-All method, which corresponds to the network diameter $D$ from Table~\ref{tab:network_types}. The diameter is known up to a scaling factor, here we report curves scaled to values comparable to AnB's execution time to ease the comparison. In fact, the intersection of same-colour curves indicates that for large networks, the AnB algorithm is asymptotically slower than the All-2-All method. This is the case for all the analyzed network topologies but the Random Geometric networks. In RG networks, All-2-All shows a steeper curve that would slow down the process for very large networks (see inset on a log-log scale). The right panel shows the fraction of residue nodes $ x = \frac{r}{N} $ in the network. Low $ x $ implies low $ r $ and hence better performance of AnB algorithm in terms of memory and communications cost (see Table~\ref{tab:generalcase}). For each network size, we report the average results for the simulation of 1,000 independent random networks. (95\% confidence intervals are reported in the left panel as shades but often are smaller than the line width.)}
  \label{fig:Scaling}
\end{figure}

\begin{table*}[tb]
	\caption{The analyzed networks. Description of the internal parameters: $m$: Number of edges with which a new node attaches to existing nodes; $p_e$: Probability of forming an edge; $ k $: Number of nearest neighbors to which the node initially connects; $p_r$: Rewiring probability; $r$: Threshold distance unto which two nodes are connected.}\vspace{0.3cm}
	\label{tab:network_types}
	\centering
        \renewcommand{\arraystretch}{1.3}
	\begin{tabular}{|c|c|c|c|c|}
          \hline
          \rowcolor{lightgray}
		Type of network & Constructing algorithm & Internal parameters & Diameter $D$ \\
		\hline
		Scale-free & Barabas\'i Albert model~\cite{Barabasi1999} & $ m = 10 $  & ~ $ D \propto \frac{\log N}{\log \log N} $, by \cite{Cohen2003}\\
		Random & Erd\"os Reny\'i model~\cite{Erdos1960} & $ p_e = \frac{20}{N} $  & $ D \propto \frac{\log N}{\log \left( p_e N \right) } $, by \cite{Chung2001} \\
		Small-world & Watts Strogatz model~\cite{Watts1998} & $ k = 20, ~ p_r = 0.5 $ & $ D \propto \log N $, by \cite{Cohen2003} \\
		Random Geometric & Algorithm by Penrose~\cite{Penrose2003} & $ r = \sqrt{\frac{10}{N}} $ & $ D \propto \frac{\sqrt{2}}{r} $, by \cite{Ganesan2018} \\
		\hline
	\end{tabular}
\end{table*}

The time, communication, and memory costs for All-2-All algorithm are relatively easy to compute. In terms of time, the algorithm ends when the messages created by every node (containing its id) reach every other node. Therefore, the time required for this to happen is equal to the diameter~$D$ of the network. In terms of communication, since each node broadcasts the id of every node to its neighborhood, the number of messages sent by each node is $N$ and hence the total number of messages sent in the whole network is $ N^2$. Finally, in terms of memory, each node needs to store the id of every node in the network. Therefore, 
the minimum memory required by each node is $N \log(N)$, by assuming 
that each id needs at least $\log(N)$ bits.

The time and communication efficiency of the ST algorithm has been outlined by Bawa et al. in~\cite{Bawa2003}. We updated their efficiency measures in order to include the changes required to allow all nodes to compute the network size. Additionally, we derived the memory cost which was not originally indicated in~\cite{Bawa2003}. The details of the complexity analysis are reported in Sec.~\secEfficiency{} of the \suppMat{}; the results are reported in Table~\ref{tab:generalcase}.

The results in Table~\ref{tab:generalcase} show that the AnB algorithm has the lowest costs in terms of memory and computation compared with the All-2-All and ST algorithms (see also Figure~\ref{fig:Scaling_Com_Mem}).
The efficiency of the AnB algoritms is higher for networks which have the number of `residue' nodes $r$ much smaller than $N$. This is the case for most random networks as shown in Figure~\ref{fig:Scaling} (right panel).  Our analysis also shows that the largest share of communication messages are typically sent by the residue nodes and the largest memory is typically required to store the ids of the residue nodes. Since the fraction of residue nodes is low for all the analyzed network classes, with the AnB algorithm the nodes send comparatively fewer messages and have lower memory requirements than with the All-2-All and ST algorithms. The only cases where the All-2-All and ST algorithms might perform better than AnB in terms of memory and communication are completely connected networks, almost completely connected networks, and networks with specific topologies (such as ring networks).  In terms of time, Figure~\ref{fig:Scaling} (left panel) shows that the All-2-All method scales as the network diameter $D$ and the AnB algorithm has comparable, or slightly worse, time performance. Finally, in terms of all three complexity aspects (time, communication, and memory), in the worst case (i.e., when $d_i = N-1$ and $r = N$), the AnB algorithm has an asymptotically complexity equal to the other algorithms (see Table~\tableWorstCase{} in the \suppMat{}). Therefore, we conclude that the AnB algorithm is advantageous for applications with constrained or high-cost communication and memory, as confirmed by the results reported in Table~\ref{tab:generalcase} and Figure~\ref{fig:Scaling_Com_Mem}.

\begin{figure}[ht]
  \centering \includegraphics[width=0.35\textwidth]{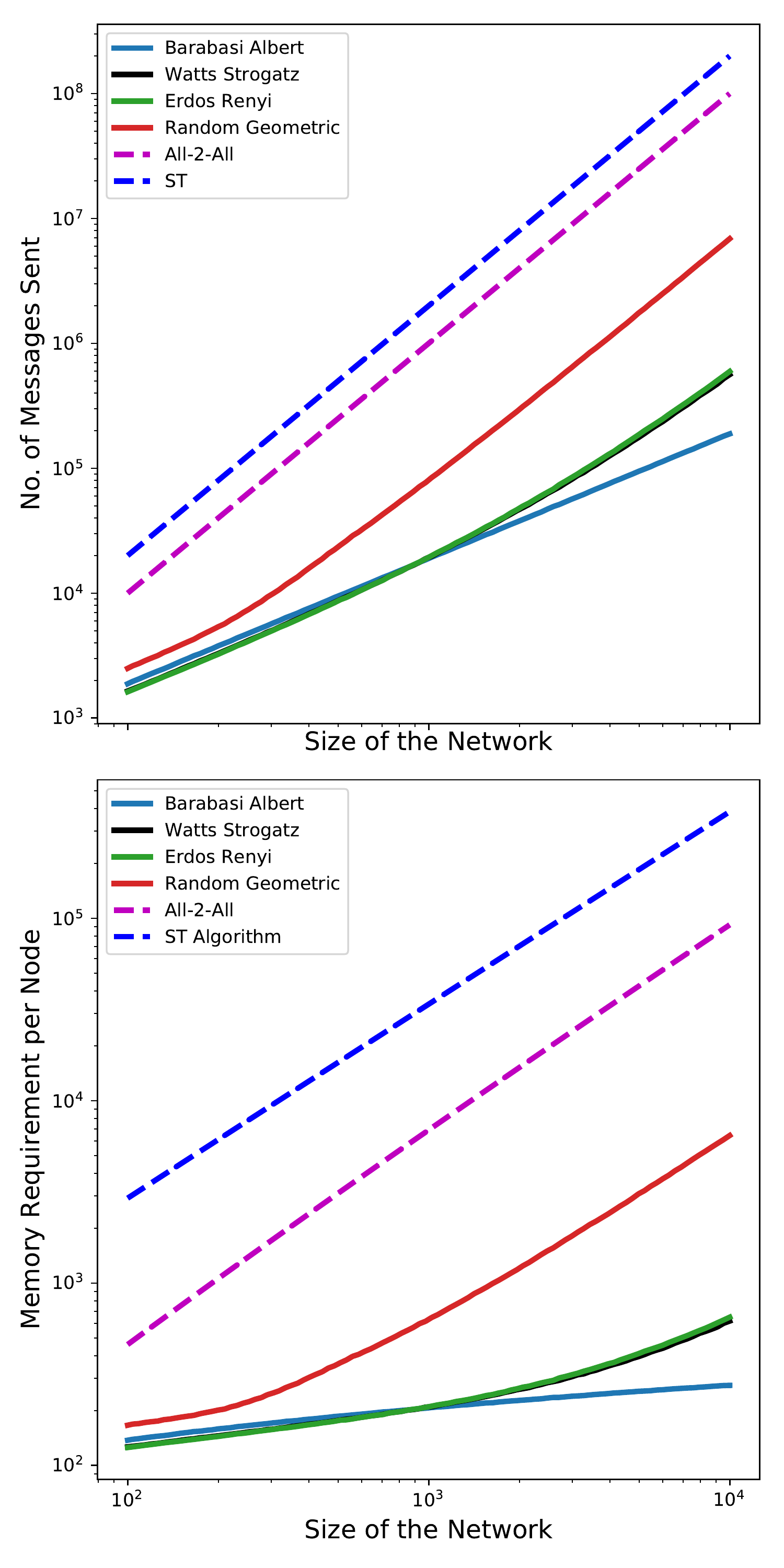}
  \caption{The AnB is the most efficient algorithm in terms of communication and memory costs, compared with the All-2-All and ST algorithms. The left panel shows the total number of messages sent by the nodes. The right panel shows the corresponding memory requirements per node with average connectivity degree $d$. In both panels, the dashed lines show the scaling for the All-2-All and ST algorithms, whereas the solid lines of various colors show the scaling for the AnB algorithm. Note that the number of messages sent and the memory requirements depends only on the network size for All-2All and ST algorithms and hence, are independent of the network topology. However, the number of messages sent and the memory requirements for AnB algorithm depends on the number of residue nodes which in turn depends on the topology of the network. Therefore, their dependence on the network topology is also explicitly shown.}
  \label{fig:Scaling_Com_Mem}
\end{figure}

\section{Conclusion}\label{sec:conclusion}

In this paper, we propose the AnB algorithm, a deterministic algorithm by which all nodes of a network can become aware of its size. The AnB algorithm assumes no inherent hierarchy among the nodes and no prior knowledge of the network topology. Instead, it depends on (a) the nodes having unique ids and (b) the nodes being able to communicate with its immediate neighbors. We also analyze the efficiency of the AnB algorithm and compare it against the known algorithms. We conclude that the AnB algorithm is significantly better than the known deterministic algorithms on average in terms of memory and communication costs. This has potential benefits in engineering where decentralized systems composed of a large number of units that operate without a central controller are spreading in various application domains, since they can offer scalable, cost-effective, robust solutions. Three examples of such domains are swarm robotics~\cite{Hamann2018}, internet of things~\cite{Atzori2010}, and wireless sensor networks~\cite{Rawat2014}.

In this concluding section, we outline some of the salient features of the AnB algorithm and the ways in which it can be extended and applied to various physical systems.

\begin{enumerate}
	
	\item \textbf{Quorum sensing:} It is notable that the local node counter $ c_i $ and the final count variable $ n_i $ are monotonic functions of time. Since both variables are aggregates of the size of the network, $ \max \left(c_i, n_i \right) $ gives a lower bound of the network size at any point in time. This can be useful in systems which are trying to determine if a quorum is present on not~\cite{Marshall2019}. Since in these cases the system is trying to determine if the network size is above a certain threshold or not, a node $ i $ can enter the broadcast phase as soon as $ c_i $ is greater than the threshold and inform the other nodes of the quorum being reached.
	
	\item \textbf{Spontaneous hierarchy creation:} While the AnB algorithm assumes no hierarchy among the nodes, the progression of the algorithm can be used to create it depending on the time when a node enters the broadcast phase. If a node enters the broadcasting phase late, it is more likely to be connected to nodes with high degrees, and hence be more `central'. Conversely, if a node enters the broadcasting phase earlier, it is more likely to be `peripheral'. While various other centrality measures exist for such classification of nodes in a network (for instance, closeness centrality~\cite{Bavelas1950} and betweenness centrality~\cite{Freeman1977}), they generally require the computation and ordering of a measure by a centralized agency. In the proposed AnB algorithm, the nodes can spontaneously organize themselves into a hierarchy.
	
	\item  \textbf{Computation of other aggregate quantities:} Similar to other previously known algorithms of network size estimation~\cite{Bawa2003, Musco2016a}, the AnB algorithm can also be used to compute other global properties across networks. For example, if each node $ i $ is associated with a property $ s_i $, they can compute the sum $ \sum s_i $ by simply setting $ c_i = s_i $ and executing the AnB algorithm. Similarly, other aggregate quantities such as averages and maximums/minimums can also be computed by suitably adopting the AnB algorithm. 
	
\end{enumerate}


\section*{Acknowledgment}

The authors acknowledge funding from the Office for Naval Research Global under grant no. 12547352 (the "Swarm Awareness" project).

\bibliographystyle{IEEEtran}
\bibliography{library}




\begin{IEEEbiography}[{
	\includegraphics[width=1in,height=1.25in,clip,keepaspectratio]{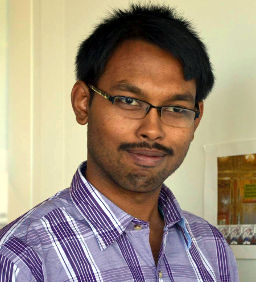}}]{Arindam Saha} is a Research Associate at the University of Sheffield. He is primarily interested in study of complex dynamical systems and has worked on various projects involving dynamics of interacting components on networks. network. He holds a MS in Theoretical Physics from IISER Kolkata, India and a PhD in Physics of Complex Systems from the University of Oldenburg. He has been working on the Swarm Awareness project since 2019.
\end{IEEEbiography}

\begin{IEEEbiography}[{
 \includegraphics[width=1in,height=1.25in,clip,keepaspectratio]{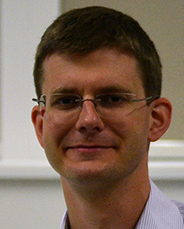}}]{James A. R. Marshall} is a Professor of Computer Science at the University of Sheffield. He leads the DiODe ERC CoG project (grant agreement no. 647704) and is co-investigator on the Swarm Awareness project. His research interests span collective behaviour, statistical decision theory, and mathematical and computational modelling of these. He is particularly interested in modelling behaviour in biological systems, and translating the results of these analyses to engineering. He is also co-founder and Chief Scientific Officer of Opteran Technologies Ltd.

\end{IEEEbiography}

\begin{IEEEbiography}[{
    \includegraphics[width=1in,height=1.25in,clip,keepaspectratio]{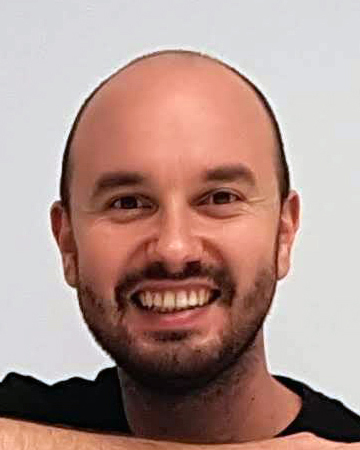}}]{Andreagiovanni Reina} is a Research Fellow in Collective Robotics at the
  University of Sheffield. He has been a member of the DiODe from 2015 to 2020 and has been the co-PI of the Swarm Awareness project since 2019. He holds a PhD in
  Applied Sciences from IRIDIA, Universit\'{e} Libre de Bruxelles,
  Belgium, and an MSc in Computer Engineering from Politecnico di
  Milano, Italy. He has been a researcher in five European projects on
  distributed robotic systems since 2009.
\end{IEEEbiography}

\appendices

\onecolumn

\section{Theorems and Proofs}
\label{sec:proof}

Let $ \mathcal{G} $ be a connected, undirected and unweighted network
whose size is to be determined by its nodes $ i $ using the AnB
algorithm. Let $ N > 0 $ be the size of $ \mathcal{G} $. Also, let
$ \mathcal{A}(t) $, $ \mathcal{B}(t) $ and $ \mathcal{C}(t) $ be the
sets of nodes in active, leaf and residue states respectively at time
$ t $. Let $ N_A(t) $, $ N_B(t) $ and $ N_C(t) $ be the number of
elements of $ \mathcal{A}(t) $, $ \mathcal{B}(t) $ and
$ \mathcal{C}(t)$, respectively. Let $ m_i(t) $ be the number of
messages of type $ h = \text{count} $ received by node $ i \in \mathcal{G} $ at
time $ t $. Finally, let $ \Gamma_i $ be the total number of neighbors
of node $ i $ and and $ \tilde{\Gamma}_i(t) $ be the number of
neighbors of the node $ i $ which are also in $ \mathcal{A}(t) $.

Let us first present the following trivial properties of these sets.

\begin{result}
	\label{result:1}
	The following statements are true for any network $ \mathcal{G} $ evolving under Algorithm~\Algorithm{}:
	\begin{enumerate}[label = (\alph*)]
		\item $ \mathcal{A}(t) $ can be partitioned as
		\begin{equation}\label{eq:A_t_Partition}
		\mathcal{A}(t) = \mathcal{A}(t+1) \cup \mathcal{B}(t+1).
		\end{equation}
		\item $ \mathcal{B}(t) $ and $ \mathcal{C}(t) $ are related as follows,
		\begin{equation}\label{eq:B_t_C_T_relation}
		\mathcal{C}(t+1) \subseteq \mathcal{B}(t).
		\end{equation}
		\item From Eqs.~\ref{eq:A_t_Partition} and \ref{eq:B_t_C_T_relation} we have,
		\begin{equation}\label{eq:C_B_A_relation}
		\mathcal{C}(t+2) \subseteq \mathcal{B}(t+1) \subseteq \mathcal{A}(t).
		\end{equation}
	\end{enumerate}
\end{result}


\begin{definition}
	Let the network $ \mathcal{G} $ said to be in resting state at time $ t \ge 0 $ if the following are true:
	\begin{equation}\label{eq:Resting_State_a}
	m_i (t) = 0 ~~~ \forall i \in \mathcal{G},
	\end{equation}
	\begin{equation}\label{eq:Resting_State_b}
	N_B(t) = 0.
	\end{equation}
	Let the ordered sequence of all times $ t \ge 0 $ when the network $ \mathcal{G} $ is in resting state be denoted by ($ T_n $) where $ n = \{0,1,2,\ldots\} $.
\end{definition}

\noindent
Using this definition, we state and prove the following results.

\begin{lemma}
	\label{lemma:1}
	The network $ \mathcal{G} $ is in resting state at time $ t = 0 $. In other words $ T_0 = 0 $.
\end{lemma}
\begin{proof}
	Prior to setting $ t=0 $ in Line~\LoopInit{}, all nodes are set to be in active state (Line~\Init{}) and they send no messages of type $ h = \text{count} $. Therefore no nodes can receive any such message at $ t = 0 $. Hence, Eqs.~\ref{eq:Resting_State_a} and \ref{eq:Resting_State_b} are satisfied at $ t = 0 $.
\end{proof}

\begin{lemma}
	\label{lemma:0}
	If $ \mathcal{G} $ is in resting state at $ t=T_n $ and $ \mathcal{A}(T_n) = \emptyset $, then $ T_{n+1} = T_n+1 $.
\end{lemma}
\begin{proof}
	The network $ \mathcal{G} $ being in resting state at $ t=T_n $ implies $ N_B(T_n) = 0 $. Hence, no message of type $ h=\text{count} $ is sent at $ t=T_n $. Therefore, $ m_i(T_n+1) = 0 ~ \forall ~ i \in \mathcal{G} $. Also, since $ \mathcal{A}(T_n) = \emptyset $ therefore, due to Eq.~\ref{eq:C_B_A_relation}, $ \mathcal{B}(T_n+1) \subseteq \mathcal{A}(T_n) = \emptyset $. Hence, the network $ \mathcal{G} $ is in the resting state at $ T_{n+1} = T_n+1 $.
\end{proof}

\begin{theorem}
	\label{theorem:1}
	For all time $ t = T_n $ when $ \mathcal{G} $ is in resting state and $ N_A(T_n) > 0 $, the following statements are true for any node $ i \in \mathcal{A}(T_n) $:
	\begin{enumerate}[label=(\alph*)]
		\item The variable $ e_i(T_n) = \tilde{\Gamma}_i(T_n) $.
		\item For each active node $ j $ in the neighborhood of $ i $, there is a corresponding element $ (j,e_j) \in \mathcal{E}_i(T_n) $ such that $ e_j = \tilde{\Gamma}_j(T_n) $.
	\end{enumerate}
\end{theorem}
\begin{proof}
	We will prove this theorem by induction.
	
	\begin{itemize}
		\item \textbf{Base case:} From Lemma~\ref{lemma:1} we know that $ \mathcal{G} $ is in resting state at $ t=0 $. Prior to setting $ t=0 $, the variable $ e_i(0) $ is initialized by counting the number of messages of type $ h=\text{echo} $ received. Since this message is sent by all nodes of the network and the message travels to all nodes in the neighborhood therefore, $ e_i(0) = \Gamma_i = \tilde{\Gamma}_i(0) $, since all nodes are in active state at $ t=0 $. Thereafter, these counts are sent to the entire neighborhood as messages of type $ h=\text{degree} $. The set $ \mathcal{E}_i(0) $ is constructed by the accumulating all such messages received. Therefore, for each $ (j,e_j)  \in \mathcal{E}_i(0) $, we have $ e_j = \Gamma_j = \tilde{\Gamma}_j(0) $.
		
		\item \textbf{Inductive step: } Let us assume that the theorem is true for some time $ t = T_k $ when $ \mathcal{G} $ is in resting state. Let $ T_{k+1} > T_k $ be the next time when $ \mathcal{G} $ is in resting state. In order to prove that the theorem is true for $ t = T_{k+1} $, we show the following in sequence:
		\begin{enumerate}[label = (\Alph*)]
			\item $ T_{k+1} > T_k + 1 $ and both parts (a) and (b) of the theorem are false at $ t = T_k + 1 $.
			\item If $ \mathcal{A}(T_k+1) = \emptyset $ then $ T_{k+1} = T_k+2 $ and the theorem is vacuously true, otherwise, if $ \mathcal{A}(T_k+1) \neq \emptyset $:
			\begin{enumerate}[label = (\roman*)]
				\item $ T_{k+1} > T_k + 2 $ and part (a) of the theorem is true and part (b) is false at $ t = T_k + 2 $.
				\item $ T_{k+1} = T_k + 3 $ and both parts (a) and (b) of the theorem are true at $ t = T_k + 3 $. 
			\end{enumerate}
			
		\end{enumerate}
		Thereby completing the inductive step.
		
		Let us consider  a node $ i \in \mathcal{A}(T_k) $ such that $ e_i = \min \{ e_j(T_k): j \in \mathcal{A}(T_k) \} $. All such nodes trivially satisfy the inequality in Line~\LeafIdentificationCondition{}. Also, since $ \mathcal{G} $ is in resting state, then $ m_i(T_k) = 0 $. Therefore, the condition stated in Line~\LeafIdentificationCondition{} is true for $ i $, which in turn implies that $ i $ sends a message of the form $ m_{i,\text{leaf}} $ at time $ T_k $ and would enter the leaf state in the next time step. Therefore,
		\begin{equation}\label{eq:N_B}
		N_B\left( T_k+1 \right) > 0.
		\end{equation}
		
		\begin{enumerate}[label=(\Alph*)]
			\item \textbf{Since $ N_B(T_{k}+1) \neq 0 $, we can say that $ T_{k+1} > T_k + 1 $.} At time $ T_k+1 $, the set $ \mathcal{A}(T_k+1) $ can be partitioned into three subsets:
			\begin{equation}\label{eq:A_partition}
			\mathcal{A}(T_k+1) = \mathcal{A}_0(T_k+1) \cup \mathcal{A}^\prime(T_k+1) \cup \mathcal{A}^{\prime\prime}(T_k+1)
			\end{equation} 
			where $ \mathcal{A}^\prime(T_k+1) $ is the set of active neighbors of nodes in $ \mathcal{B}(T_k+1) $, $ \mathcal{A}^{\prime \prime}(T_k+1) $ is the set of active neighbors of nodes in $ \mathcal{A}^\prime(T_k+1) $ which are not already in $ \mathcal{A}^\prime(T_k+1) $, and $ \mathcal{A}_0(T_k+1) $ is the set of all other active nodes in $ \mathcal{A}(T_k+1) $. 
			
			Consider the value of $ e_l(T_k+1) ~ \forall ~ l \in \mathcal{A}^\prime (T_k+1) $. It has not yet changed since $ t = T_{k} $. However, the actual value of $ \tilde{\Gamma}_l(T_k+1) $ has changed. Therefore, \textbf{part (a) of the lemma is false at $ t = T_k+1 $.}
			
			Similarly, consider the elements of $ \mathcal{E}_m(T_k+1)~ \forall ~ m \in \mathcal{A}^{\prime \prime}(T_k+1) $. Since each element in $ \mathcal{A}^{\prime \prime} $ has an active neighbor in $ \mathcal{A}^\prime $, there exists at least one $ (l, e_l) \in \mathcal{E}_m(T_k+1)  $ such that $ l \in \mathcal{A}^\prime $. Since this $ e_l $ does not represent the current value of $ \tilde{\Gamma}_l $ \textbf{therefore, part (b) of the lemma is also false at $ t = T_k+1 $.}
			
			Note that, the active neighbors of all nodes $ p \in \mathcal{A}_0(T_k+1) $ are in the set $ \mathcal{A}^{\prime \prime}(T_k+1) $. Since, nodes in $ \mathcal{A}^{\prime \prime}(T_k+1) $ are still in the active state, therefore, $ e_p = \tilde{\Gamma}_p(T_k+1) $ for all $ p \in \mathcal{A}_0(T_k+1) $. Furthermore, the neighbors of nodes in $ \mathcal{A}^{\prime \prime}(T_k+1) $ are in $ \mathcal{A}^{\prime}(T_k+1) $. Since all nodes in $ \mathcal{A}^{\prime}(T_k+1) $ are in active state, the number of active neighbors of nodes in $ \mathcal{A}^{\prime \prime}(T_k+1) $ has not changed. Therefore each element $ (m,e_m) \in \mathcal{E}_p(T_k+1) $ still accurately represents the active neighbors $ m $ of $ p $ and the number of active neighbors $ e_m $ of $ m $.  
			
			The number of messages of type $ h=\text{leaf} $ received by a node $ b \in \mathcal{B}(T_k+1) $ gives the number of its neighbors which were active at $ t=T_k $ but have transitioned to the leaf state at $ t=T_k+1 $ . Note that $ e_b(T_k) = \Gamma_b(T_k) $. Since $ e_b(T_k+1) $ is obtained by taking the difference between $ e_b(T_k) $ and the aforementioned number of messages therefore, it gives the number of neighbors of $ b $ which are still in the active state at $ t = T_k+1 $. In other words,
			\begin{equation}\label{eq:tilde_n_e_b_relation}
			e_b(T_k+1) = \tilde{\Gamma}_b(T_k+1)
			\end{equation}
			If $ e_b(T_k+1) = 0 $, $ b $ enters the residue state at $ t = T_k+1 $; otherwise, it sends a message $ m_{b,\text{count}} $ and enters the inactive state.
			
			Finally, since no message of type $ h = \text{count} $ or $ h = \text{reduce} $ are received by any node, we can say that $ e_i(T_k+1) = e_i(T_k) $ and $ \mathcal{E}_i(T_k+1) = \mathcal{E}_i(T_k) $. Therefore, no active nodes become leaf nodes. Hence,
			\begin{equation}\label{eq:A_change_1}
			\mathcal{A}(T_k+2) = \mathcal{A}(T_k+1)
			\end{equation}
			and
			\begin{equation}\label{eq:B_change_1}
			\mathcal{B}(T_k+2) = \emptyset.
			\end{equation}
			
                      \item We can assume two mutually exclusive cases: $ \mathcal{A}(T_k+1) = \emptyset $ or $ \mathcal{A}(T_k+1) \neq \emptyset $.

                        If we assume $ \mathcal{A}(T_k+1) = \emptyset $, from Eq.~\ref{eq:A_t_Partition}, $ \mathcal{A}(T_k+1) = \emptyset $ implies $ \mathcal{B}(T_k+2) = \emptyset $. Also from Eq.~\ref{eq:A_partition}, $ \mathcal{A}(T_k+1) = \emptyset $ implies $ \mathcal{A}^\prime(T_k+1) = \emptyset $. Therefore, from Eq.~\ref{eq:tilde_n_e_b_relation}, we have $ e_b(T_k+1) = 0 ~ \forall ~ b \in \mathcal{B}(T_k+1) $ which in turn implies $ m_i(T_k+2) = 0 $. Hence, $ \mathcal{A}(T_k+1) = \emptyset $ implies that the network is in resting state at $ T_{k+1} = T_k+2 $.
			From Eq.~\ref{eq:A_change_1} we have, $ \mathcal{A}(T_{k+1}) = \mathcal{A}(T_k+2) = \mathcal{A}(T_k+1) = \emptyset $ also implies that the hypothesis of the theorem is false at time $ t = T_{k+1} $. Therefore, the statement of the theorem is vacuously true.
			
			For the remaining part of the proof, we assume that,
			\begin{equation}\label{eq:assumption}
			\mathcal{A}(T_k+1) \neq \emptyset.
			\end{equation}
			
			\begin{enumerate}[label = (\roman*)]
				\item The messages $ m_{b,\text{count}} $ are received by all neighbors of $ b \in \mathcal{B}(T_k+1) $ at $ t = T_k+2 $. \textbf{This proves that the network is not in the resting state at $ t = T_k+2 $.} On receiving each such message, the active nodes $ l \in \mathcal{A}^\prime(T_k+1) $ remove the element $ (b,e_b) $ from $ \mathcal{E}_l(T_k+2) $ and decrease the value of $ e_l(T_k+2) $ by one. The total number of such messages received by $ l $ at $ t = T_k+2 $ is equal to reduction in $ \tilde{\Gamma}_l $ from $ t = T_k $ to $ t = T_k+1 $. Furthermore, since $ \mathcal{A}(T_k+2) = \mathcal{A}(T_k+1) $, therefore, $ \tilde{\Gamma}_l $ has not changed from $ t=T_k+1 $ to $ t=T_k+2 $. Therefore, $ e_l(T_k+2) = \tilde{\Gamma}_l(T_k+2) ~ \forall  l \in A^{\prime}(T_k+1) $. The $ e_{l^\prime}(T_k+2) $ of all other nodes $ l^\prime \in \mathcal{A}(T_k+3) \setminus A^{\prime}(T_k+1) $ gives the actual number of their active neighbors anyways. \textbf{Therefore, part (a) of the theorem is true at $ t = T_k+2 $}.
				
				However, the nodes $ m \in \mathcal{A}^{\prime \prime}(T_k+1) $ remain unaware of the changes in number of active neighbors of the nodes in $ \mathcal{A}^\prime(T_k+1) $. \textbf{Hence, part (b) of the theorem is still false at $ t=T_k+2 $} because the value of $ e_l $ stored in $ (l,e_l) \in \mathcal{E}_m(T_k+2) $ is still equal to $ e_l(T_k) $ and not to $ e_l(T_k+2) $. In order to notify the nodes $ m \in \mathcal{A}^{\prime \prime}(T_k+1) $ of the changes of $ e_l(T_k) $, nodes  $ l \in \mathcal{A}^\prime(T_k+2) $ send a message $ m_{l,\text{reduce}} $ for each $ m_{b,\text{count}} $ received.
				
				Finally, since any node $ i \in \mathcal{A}(T_k+2) $ which has changed $ e_i(T_k+1) $ has also received a message of type $ h=\text{count} $, it does not change its state. Any other node in $ \mathcal{A}(T_k+2) $ can also not change its state due to arguments similar to the ones presented in (A). Therefore,
				\begin{equation}\label{eq:A_change_2}
				\mathcal{A}(T_k+3) = \mathcal{A}(T_k+2)
				\end{equation}
				and
				\begin{equation}\label{eq:B_change_2}
				\mathcal{B}(T_k+3) = \emptyset.
				\end{equation}

				\item Eq.~\ref{eq:B_change_1} implies that there were no messages $ m_{b,\text{count}} $ sent at time $ T_k+2 $. Therefore, 
				\begin{equation}\label{eq:m_i_T_k+3}
				m_i(T_k+3) = 0 ~ \forall ~ i \in \mathcal{G}.
				\end{equation}
				\textbf{This, when combined with Eq.~\ref{eq:B_change_2} implies $ \mathcal{G} $ is in the resting state at $ T_k+3 $.}
				
				Eq.~\ref{eq:m_i_T_k+3} also implies that $ e_i(T_k+3) = e_i(T_k+2) ~ \forall ~ i \in \mathcal{G} $. \textbf{Therefore, part (a) of the theorem still holds at $ t=T_k+3 $.}
				
				Now, let us consider the messages $ m_{l,\text{reduce}} $ sent by all nodes $ l \in \mathcal{A}^\prime(T_k+1) $. On receiving each such message, the nodes $ m \in \mathcal{A}^{\prime\prime}(T_k+1) $ update $ (l,e_l) \in \mathcal{E}_m(T_k+2) $ to $ (l,e_l-1) $. Since the number of messages received equals the reduction in the degree of node $ l $, after receiving all the messages $ m_{l,\text{reduce}} $, the tuple $(l,e_l)$ gives the correct number $e_l$ of active neighbors of $ l $. \textbf{Therefore, part (b) of the theorem is satisfied.}
			\end{enumerate}
		\end{enumerate}
	\end{itemize}
	Hence, if we assume the Theorem~\ref{theorem:1} to be true at some $ t=T_k $, it is also true for another $ T_{k+1} > T_k $. We have $ T_{k+1} = T_k+2 $ if there are no active nodes left at $ T_k+1 $ (i.e. $\mathcal{A}(T_k+1) = \emptyset$); otherwise, if $\mathcal{A}(T_k+1) = \emptyset$, we have $ T_{k+1} = T_k+3 $. Since we already know that the theorem is true for $ t=T_0=0 $, it is true for all $ T_n $ by induction.
\end{proof}

\noindent
From Lemma~\ref{lemma:0} and Theorem~\ref{theorem:1}, we deduce the following results for resting times $ T_n $.

\begin{result}
	\label{result:2}
	For any $ T_n $,
	\begin{equation}\label{eq:T_k_evovution}
	T_{n+1} \le T_n + 3.
	\end{equation}
\end{result}

\begin{result}
	\label{result:3}
	If $ N_A(T_n)  = 0 $, then
	\begin{equation}\label{eq:N_A_non_evolution}
	N_A(T_{n+1}) = N_A(T_n) = 0.
	\end{equation}
	Instead, if $ N_A(T_n)  > 0 $, then
	\begin{equation}\label{eq:N_A_evolution}
	N_A(T_{n+1}) < N_A(T_n).
	\end{equation}
\end{result}
\begin{result}
	\label{result:4}
	If $ \mathcal{A}(T_k) \neq \emptyset $, then
	\begin{equation}\label{eq:A_subset}
	\mathcal{A}(T_{k+1}) = \mathcal{A}(T_{k}+1) \subset \mathcal{A}(T_k).
	\end{equation}
\end{result}

\noindent
It directly follows from Lemma~\ref{lemma:1} and Result~\ref{result:2} that,
\begin{equation}\label{eq:T_n_form}
T_n \le 3n \,.
\end{equation}
Additionally, the following corollary follows directly from Result~\ref{result:2} and Eq.~\ref{eq:T_n_form}.

\begin{corollary}
	\label{corollary:1}
	For a network $ \mathcal{G} $ of size $ N > 0 $ evolving under Algorithm~\Algorithm{}, there exists a time $ t_R \le 3N $ such that $ N_A(t_R) = 0 $.
\end{corollary}
\begin{proof}
	Let us assume that $ N_A(t) > 0 $ for all $ t \le 3N $. This implies,
	\begin{equation}
	N_A(3N) > 0
	\end{equation}
	Since $ N_A(t) $ is a monotonically decreasing function of time and $ T_N \le 3N $ (due to Eq.~\ref{eq:T_n_form}), we can say that,
	\begin{equation}\label{eq:Contradiction_a}
	N_A(T_N) > 0.
	\end{equation}

        \noindent
	Note that $ N_A(t) $ can change only is steps of one since it is a non-negative integer function. Therefore, using Eq.~\ref{eq:N_A_evolution},
	\begin{equation}\label{eq:Contradiction_aux}
	N_A(T_{N}) \le N_A(T_{N-i}) - i
	\end{equation}
	for any positive integer $ i \le N $. Setting $ i = N $, we get,
	\begin{equation}\label{eq:Contradiction_b}
	N_A(T_N) \le 0
	\end{equation}
	since $ T_0 = 0 $ and $ N_A(0) = N $ due to all nodes of the network being active at $ t=0 $. This is a direct contradiction to Eq.~\ref{eq:Contradiction_a}. Therefore, our assumption was wrong, hence proving the corollary.
\end{proof}

\begin{corollary}
	\label{corollary:2}
	The number of residue nodes $ N_C(t) \ge 0 $ only if $ t = T_n+2 $ for some $ n $ and $ N_A(T_n) > 0 $. Otherwise, $ N_C(t) = 0 $.
\end{corollary}
\begin{proof}
	It directly follows from Eqs.~\ref{eq:N_B}, \ref{eq:B_change_1} and \ref{eq:B_change_2} that if $ N_A(T_n) > 0 $, nodes in leaf state are present only at $ t=T_n+1 $. This, in turn, implies $ N_C(t) \ge 0 $ only if $ t = T_n+2 $ due to Eq.~\ref{eq:C_B_A_relation}. If there are no active nodes present at $ T_n $, then it is obvious that there would be no leaf or residue nodes in the future.
\end{proof}

\begin{definition}
	Let the cumulative residue set of a network $ \mathcal{G} $ at time $ t $ be defined as,
	\begin{equation}\label{eq:P_definition}
	\mathcal{P}(t) = \bigcup_{i=0}^{t} \mathcal{C}(i).
	\end{equation}
\end{definition}

\begin{definition}
  Let the cumulative residue index of a network $ \mathcal{G} $ at time $ t $ be defined as the number of elements in $ \mathcal{P}(t) $,
  \begin{equation}\label{eq:r_definition}
    r(t) = \left| \mathcal{P}(t) \right|.
  \end{equation}
\end{definition}
\noindent
Clearly, $ r(t) $ is a non-negative, non-decreasing integer function of $ t $. We now prove the following corollary regarding $ r(t) $.

\begin{corollary}
	\label{corollary:3}
	If $ T_R $ is the time when the network $ \mathcal{G} $ is the resting state, $ N_A(T_R) = 0 $, and $ N_A(T_{R-1}) > 0 $, then $ r(T_R) > 0 $.
\end{corollary}
\begin{proof}
	If $ r(T_{R-1}) > 0 $, then the corollary is trivially proved. Therefore, let us consider the case when $ r(T_{R-1}) = 0 $. Since $ N_A(T_{R-1}) > 0 $, we can use Result~\ref{result:4} to get $ \mathcal{A}(T_{R-1}+1) = \mathcal{A}(T_R) = \emptyset $. This in turn implies $ \mathcal{B}(T_{R-1}+1) = \mathcal{A}(T_{R-1}) $ due to Eq.~\ref{eq:A_t_Partition}.
	This leads to the situation where, at $ t=T_{R-1}+1 $ none of the leaf nodes have any active neighbors. Therefore, due to Eq.~\ref{eq:tilde_n_e_b_relation},
	\begin{equation}\label{eq:final_leaves}
	e_b(T_{R-1}+1) = 0 ~~~ \forall b \in \mathcal{B}(T_{R-1}+1).
	\end{equation}

\noindent
	Hence, the condition in Line~\StartTransitionCondition{} of the Algorithm~\Algorithm{} is satisfied for all nodes $ b \in \mathcal{B}(T_{R-1}+1) $, which become residue nodes. Thus, from Theorem \ref{theorem:1}, because $N_A(T_R) = 0$, we have
	\begin{equation}\label{eq:Last_step}
	\mathcal{C}(T_{R-1}+2) = \mathcal{C}(T_{R}) = \mathcal{B}(T_{R-1}+1) = \mathcal{A}(T_{R-1})
	\end{equation}
	Hence, $ \mathcal{C}(T_R) \neq \emptyset $ and therefore $ r(T_R) > 0 $.
\end{proof}

\begin{definition}
	Let the overall count $ I(t) $ be defined as,
	\begin{equation}\label{eq:I_definition}
	I(t) = \sum_{i \in \mathcal{Q}(t)} c_i(t)
	\end{equation}
	where $ \mathcal{Q}(t) = \mathcal{A}(t) \cup \mathcal{P}(t) $.
\end{definition}

\begin{theorem}
	\label{theorem:2}
	For any time $ T_n $ when $ \mathcal{G} $ is in resting state,
	\begin{equation}\label{eq:Invariant}
	I(T_n) = N
	\end{equation}
	when $ N $ is the size of $ \mathcal{G} $.
\end{theorem}
\begin{proof}
	
	We prove this by induction. At $ t=T_0=0 $, the set $ \mathcal{A}(t) $ contains all nodes of the network. Since $ c_i(t) $ is initialized as $ c_i(0)=1 ~ \forall ~ i \in \mathcal{G} $, therefore, $ \mathcal{Q}(T_0) = \mathcal{Q}(0) = N $.
	
	Now, let us assume that the theorem is true for some time $ t=T_k $ such that $ \mathcal{A}(T_k) \neq \emptyset $. Now, by definition,
	\begin{equation}
	\mathcal{Q}(T_k) = \mathcal{A}(T_k) \cup \mathcal{P}(T_k).
	\end{equation}
	Now, since $ \mathcal{A}(T_k) $ and $ \mathcal{P}(T_k) $ are disjoint sets, we have,
	\begin{equation}\label{eq:Induction_a}
	I(T_k) = \sum_{a \in \mathcal{A}(T_k)} c_a(T_k) + \sum_{p \in \mathcal{P}(T_k)} c_p(T_k) = N.
	\end{equation}
	
	Let us now consider the network at $ t=T_{k+1} $. From Eqs.~\ref{eq:A_t_Partition} and \ref{eq:A_subset} we have,
	\begin{equation}\label{eq:Induction_b}
	\mathcal{A}(T_k) = \mathcal{A}(T_k+1) = \mathcal{A}(T_k) \setminus \mathcal{B}(T_k+1).
	\end{equation}
	Also due to Corollary~\ref{corollary:3}, we have,
	\begin{equation}
	\mathcal{P}(T_{k+1}) = \mathcal{P}(T_k) \cup \mathcal{C}(T_k+2).
	\end{equation}
	Therefore,
	\begin{equation}
	\mathcal{Q}(T_{k+1}) = \mathcal{A}(T_{k+1}) \cup \mathcal{P}(T_{k+1}) = \left[ \mathcal{A}(T_k) \setminus \mathcal{B}(T_k+1) \right] \cup \left[ \mathcal{C}(T_k+2) \cup \mathcal{P}(T_k) \right].
	\end{equation}
	Hence,
	\begin{equation}
	\mathcal{Q}(T_{k+1}) = \left[ \left( \mathcal{A}_0(T_k) \cup \mathcal{A}^\prime(T_k+1) \right) \setminus \left( \mathcal{B}^\prime(T_k+1) \cup \mathcal{C}(T_k+2) \right) \right] \cup \left[ \mathcal{C}(T_k+2) \cup \mathcal{P}(T_k) \right]
	\end{equation}
	where, $ \mathcal{A}^\prime(T_k+1) $ is the set of active neighbors at $ t=T_k+1 $, $ \mathcal{A}_0(T_k) $ is the set of the rest of the nodes in $ \mathcal{A}(T_k) $ and $ \mathcal{B}^\prime(T_k+1) $ is the set of nodes in $ \mathcal{B}(T_k+1) $ which did not become residue nodes.
	
	Now $ I(T_{k+1}) $ is the sum of $ c_i(T_{k+1}) $ for all nodes $ i \in \mathcal{Q}(T_{k+1}) $. Noting that the fourth and the fifth sets in the right hand side of the previous equations are identical, we have,
	\begin{equation}
	I (T_{k+1}) = \sum_{a_0 \in \mathcal{A}_0(T_k)} c_{a_0}(T_{k+1}) + \sum_{a^\prime \in \mathcal{A}^\prime(T_k+1)} c_{a^\prime}(T_{k+1}) - \sum_{b^\prime \in \mathcal{B}^\prime(T_k+1)} c_{b^\prime}(T_{k+1}) + \sum_{p \in \mathcal{P}(T_k)} c_{p}(T_{k+1}).
	\end{equation}
	Since $ c_i(t) $ does not change unless a message of type $ h=\text{count} $ is received, we have,
	\begin{equation}
	I(T_{k+1}) = \sum_{a_0 \in \mathcal{A}_0(T_k)} c_{a_0}(T_{k}) + \sum_{a^\prime \in \mathcal{A}^\prime(T_k+1)} \left( c_{a^\prime}(T_{k}) + \delta_{a^\prime} \right) - \sum_{b^\prime \in \mathcal{B}^\prime(T_k+1)} c_{b^\prime}(T_{k}) + \sum_{p \in \mathcal{P}(T_k)} c_{p}(T_{k})
	\end{equation}
	which can be simplified using Eq.~\ref{eq:Induction_a} into,
	\begin{equation}
	I(T_{k+1}) = I(T_k) + \left( \sum_{a^\prime \in \mathcal{A}^\prime(T_k+1)} \delta_{a^\prime} - \sum_{b^\prime \in \mathcal{B}^\prime(T_k+1)} c_{b^\prime}(T_k) \right).
	\end{equation}
	
	Now,
	\begin{align}\label{eq:proof_a}
	\sum_{a^\prime \in \mathcal{A}^\prime(T_k+1)} \delta_{a^\prime} &=
		\sum_{\substack{a^\prime \in \mathcal{A}^\prime(T_k+1) \\ b^\prime \in \mathcal{B}^\prime(T_k+1)\\ a^\prime \text{ is connected to } b^\prime}} \frac{c_{b^\prime} (T_k+1)}{e_{b^\prime}(T_k+1)} \\
		~&~ \nonumber\\
		~ &= 
	\sum_{b^\prime \in \mathcal{B}^\prime(T_k+1)} \tilde{n}_{b^\prime} \left( \frac{c_{b^\prime} (T_k)}{e_{b^\prime}(T_k+1)} \right) \\
		~&~ \nonumber\\
		~ &=
	\sum_{b^\prime \in \mathcal{B}^\prime(T_k+1)} c_{b^\prime}(T_k)
	\end{align}
	due to Eq.~\ref{eq:tilde_n_e_b_relation}. Therefore,
	\begin{equation}\label{eq:Induction_c}
	I(T_{k+1}) = I(T_k) = N.
	\end{equation}
	This proves that the theorem is true for all $ T_n \le T_R $ such that $ \mathcal{A}(T_{R-1}) \neq \emptyset $ and $ \mathcal{A}(T_R) = \emptyset $. For any $ tT_n >T_R $, the overall count $ I(T_n) $ still remains invariant, since $ \mathcal{A}(T_n) $ continues to be an empty set and hence $ \mathcal{P}(T_n) $ does not change further.
\end{proof}

\begin{theorem}
	\label{theorem:3}
	There exists some time $ t_{\text{max}} \le 4N+1 $ such that for all $ t \ge t_0 $,
	\begin{equation}\label{key}
	n_i(t) = N ~~~ \forall ~ i \in \mathcal{G}.
	\end{equation}
\end{theorem}
\begin{proof}
	We know from Corollary~\ref{corollary:1} that for some minimum time $ t_r \le 3N $, the number of active nodes in the network $ \mathcal{G} $ becomes zero. Therefore $ \mathcal{A}(t_r) = 0 $. Hence, by $ t = T_{R+1} = t_r+2 $ all messages of the form $ m_{p,\text{broadcast}} $ have been sent by all nodes $ p \in \mathcal{P}(T_{R+1}) $. The time required by any such message to reach any other node of the network is $ t_b \le N-1 $. Therefore, by using Theorem~\ref{theorem:2}, we can say that, at time $ t \ge t_r+t_b \le 4N+1 $, the final count variable for each node $ i \in \mathcal{G} $ is,
	\begin{equation}\label{eq:final_proof}
	n_i(t) = \sum_{p \in \mathcal{P}(t)} c_p(t) = I(t) = N.
	\end{equation}
	
\end{proof}

\newpage

\section{Complexity Analysis}
\label{sec:Appendix}

In this section, we provide the details of the efficiency of the AnB
and the ST algorithms with respect to time, communication, and memory costs.

\subsection{Time Cost}
\label{subsec:Time_cost}

It is difficult to make theoretical estimates about the number of time steps that it takes for the AnB algorithm to work. This is because the number of nodes getting changing states from active to leaf state at each time step depends on the topology of the network. However, the topology of the `remaining' network evolves as the algorithm progresses as a result of nodes entering into the inactive state. Therefore, while it is possible to estimate the fraction of nodes which get eliminated at the first iteration of the algorithm, estimating the fraction in all subsequent iterations of network reduction is difficult due to difficulties in gauging the changes in the degree distribution and topology of the remaining network. Nevertheless, we can obtain insights into the time costs of the AnB algorithm by analyzing the results of numerical simulations.

The AnB algorithm can be divided into two distinct phases from the point of view of the network: (a) The network reduction phase: where there are active or leaf nodes still present in the network and the information about the size of the network is being concentrated into a few residue nodes; and (b) The broadcast phase: where no active or leaf nodes are present in the network and the concentrated information is broadcast to all other nodes of the network. Furthermore, since the algorithm successively eliminates nodes with low degrees, the residue nodes which remain after the elimination process are more probable to be nodes with a high degree. In other words, the number of iterations taken in the `active' and `inactive' phases depends on the distribution of high/low degree nodes which are\textbf{} determined by the network topology.

Consider a network with a heterogeneous degree distribution such as the Barab\'asi-Albert network, which has few nodes with extremely high degrees and numerous nodes with low degrees. This skewness in degree distribution means that most nodes with low degrees get eliminated without going through the residue state, whereas the few nodes with high degrees become residues. This results in the network having few residue nodes which are well-connected. This makes the broadcast of the counts in residue nodes more efficient. In contrast, consider a network with approximately homogeneous distribution such as the Random Geometric network. Here, since the degrees of all nodes are approximately the same, there is a high probability that there will be a greater number of nodes which go through the residue state.

Having few well connected nodes reduces the number of iterations required in the broadcast phase as the number of residue messages to be sent across the network is low and can be broadcast faster due to high connectivity. That is why networks with more heterogeneous degree distribution, such as the Erd\"os-R\'enyi, Barab\'asi-Albert and Watts-Strogatz networks, spend a lower fraction of time in the broadcast phase. In contrast, the Random Geometric network with a more homogeneous degree distribution spends a higher proportion of iterations in the broadcast phase.

Our numerical simulations confirm our insights and show a lower convergence time for heterogeneous networks (see Figure~\Figure{} top panel). Our results also show that the time spent in the network reduction phase is also significantly lower in networks with heterogeneous degree distribution. This is probably due to the fact that an heterogeneous degree distribution allows for a greater number of nodes to be eliminated in one iteration.

Finally, the numerical simulations performed on the different types of
random networks (see Figure~\Figure{} top panel) reveal that
the time taken by the algorithm scales better than $ \log N $ for
Erd\"os-Reny\'i, Barabas\'i-Albert and Watts-Strogatz networks. This
is evident from the sub-linear nature of the plots. On the Random
Geometric networks, the algorithm scales worse than $ \log N
$. Instead, as the broadcast-time required for the network increases
relative to the network reduction-time for large networks, the total
time required scales as a power of network size $ N $.


\subsubsection{Time Cost of the ST Algorithm}

The time cost of the ST algorithm when the single node $i$ computes
the size of the network depends on the topological location of node
$i$. When all nodes need to compute the ST algorithm the time
necessary is exactly $2D$, where $D$ is the network diameter. This is
the time necessary to let a message go back and forth throughout the
entire network. Therefore the asymptotic worst-case time complexity
for the ST algorithm is $\mathcal{O}(D)$.

\subsection{Communication Cost}
\label{subsec:Communication_cost}

To assess the communication cost we compute the expected number of messages to reach convergence, \emph{i.e.} all nodes have the variable $ n_i $ equal to the network size. In the proposed algorithm the nodes send various types of messages at various stages of the algorithm. Note that, we assume each message to be `broadcast' to the neighbors rather than multicast. Therefore, whenever a message is sent from a node to all its neighbors, we count it as a single message. We count the number of messages sent at each stage of the algorithm as follows.

\begin{itemize}
	
	\item \textbf{Initial count of all neighbors (message type $ h = \text{echo} $):} In the initiation phase of the algorithm, each node announces its presence to all its neighbors so that each node becomes aware of its neighborhood. This message is sent once by each node in the network. Therefore, the total number of messages sent in this stage is $M_1 = N$.
	
	\item \textbf{Broadcasting the initial number of neighbors (message type $ h = \text{degree} $):} Each node then broadcasts its number of neighbors. Since this message is also sent once by each node of the network, the total number of messages sent in this stage is $M_2 = N$.
	
	\item \textbf{Declaring transition to leaf state (message type $ h = \text{leaf} $):} When each node changes its state to $s_i = L$, it broadcasts a message so that any neighboring node in leaf state may update its effective degree. Since the transition from $s_i = A$ to $s_i = L$ is made once by each node, the total number of messages sent in this stage is $M_3 = N$.
	
	\item \textbf{Declaring transition to inactive state (message type $ h = \text{count} $):} Having a network with $ r $ residue nodes, the transition from $s_i = L$ to $s_i = I$ is made by $N - r$ nodes. Therefore, the number of  messages sent in this stage is $M_4 = N - r$.
	
	\item \textbf{Updating effective degree (message type $ h = \text{reduce} $):} When a node $i$ in state $s_i = A$ receives a message informing the transition of node $j$ from state $s_j = L$ to $s_j = B$, its effective degree $ e_i $ changes. Then, node $i$ has to send its updated effective degree. While the exact number of messages sent informing the changes in effective degree depends on the topology of the network, we can compute its upper bound to be $M_5 = Nd$, where $d$ is the average degree of the nodes in the network. The reasoning behind this result is as follows. 
	
	
	Consider a node $j$ which is transitioning from $s_j = L$ to $s_j = I$. This will lead to a change in effective degree of all its neighbors. Out of these neighbors, only the nodes in active state send a message informing the change in effective degree. Therefore the number of update messages sent is equal to the number of edges between the node $j$ and its active neighbors. It also follows that no update message would be sent by the node $j$ after it has transitioned to the inactive state. Therefore, for the purpose of counting the number of update messages, we must iteratively `remove' the nodes which have transitioned to the inactive state along with all their edges. The total number of edges removed thus would give the upper bound of the number of update messages sent. Since the maximum number of edges in the network is $Nd$, therefore the maximum number of update messages is $M_5 = Nd$.
	
	\item \textbf{Broadcasting messages from the residue nodes (message type $ h = \text{broadcast} $):} In the final stage of the algorithm, each node in the residue state creates a broadcast message which is then broadcast throughout the network. If there are $r$ nodes reaching the residue state, the number of messages sent in this stage is $M_6 = Nr$.
	
\end{itemize}

The upper bound of the total number of messages sent in the entire duration of the algorithm is, therefore,
\begin{equation}
M = \sum_{k = 1}^6 M_k = 4N - Nx + Nd + N^2x
\end{equation}
where,
\begin{equation}
x = \frac{r}{N}
\end{equation}
is the fraction of residue nodes in the network.

In the All-2-All broadcast method for node counting, each of the $N$ nodes sends its id to all the $N$ nodes of the network. Therefore, the total number of messages sent in the algorithm is $M_0 = N^2$.

Comparing the two algorithms in terms of the communication costs, we can say that the proposed algorithm is better if
\begin{equation*}
4N - Nx + Nd + N^2x < N^2 
\end{equation*}
or equivalently if
\begin{equation}
d < N(1-x) + x - 4.
\label{eq:communications_condition}
\end{equation}

In other words, the proposed algorithm is better than the All-2-All broadcast method if the average degree of the nodes is less than the threshold on the right-hand side of Eq.~\ref{eq:communications_condition}. This threshold depends on the size of the network $ N $ and the fraction of residue nodes $ x $. Since $ 0 < x \leq 1 $, we now analyze Eq.~\ref{eq:communications_condition} in the limiting cases.

If $ x = 1 $, all the nodes of the network have gone through the residue state. This occurs in the special case when all nodes of the network have the same degree (in other words, we have a regular network). In such a limiting case, Eq.~\ref{eq:communications_condition} reduces to $ d < -3 $, which is impossible. Therefore, in regular networks of any size the direct broadcast method is better than the proposed algorithm in terms of communication costs. On the other hand, if $ x \rightarrow 0 $, Eq.~\ref{eq:communications_condition} becomes $ d < N - 4 $. Since in a typical network, the average degree of a node is much less than the number of nodes, we can say that the proposed method is better for almost any network where the fraction of residues is close to zero.

When applied to random graphs of aforementioned types in simulations, we observe that the fraction of residue nodes $ x $ comes out to be close to zero for sufficiently large networks (see Figure~\Figure{} bottom panel). For instance, Barab\'asi Albert, Erd\"os R\'enyi and Watts-Strogatz networks of size 10,000 have the fraction of residue nodes $ x $ below 0.02. This implies that the threshold average degree $ d $ for which the proposed algorithm outperforms the All-2-All broadcast method is $ d > 9796 $. Even for a Random Geometric network, where $ x \approx 0.07 $, the threshold is approximately $d \approx 9300 $. Since the approximate average degrees for the networks is approximately $ d \approx 20 $ to $ 30 $, which is significantly less than the indicated thresholds, we can conclude that the proposed algorithm has significantly lower communication costs than the All-2-All broadcast method. Moreover, since Figure~\Figure{} indicates that $ x $ remains constant, or decreases, as the network size increases, we expect that the proposed algorithm would be even more efficient for larger networks.

\subsubsection{Communication Cost of the ST Algorithm}

The number of messages sent across the network is at least $2N^2$.
More specifically, each node sends at least two messages for each
query it receives: one to establish the hierarchy during the tree
construction phase and one to send towards the root the count of the
nodes. Since each node computes the ST algorithm independently, each
node receives $N$ queries, and therefore it sends at least $2N$
messages. We can thus derive that the total number of messages sent
across the network is at least $2N^2$.
The asymptotic worst-case complexity in terms of communication for the
ST algorithm is therefore $\mathcal{O}(N^2)$.

\subsection{Memory Cost}
\label{subsec:Memory_cost}

Throughout the execution of the algorithm, each node keeps track of a number of internal variables. The state variable $ s_i $ can take one of four different values and therefore, has memory requirements independent of the network properties. Since these memory requirements are relatively small, we ignore them in the further analysis. 

The variables $ c_i $, $ n_i $ and $ e_i $ can be numbers up to and including $ N $. Hence, the memory requirement for each of them is proportional to $ \log N $. Similarly, the memory requirement for any message variable $ m_{i,h} $ can scale as $ 2 \log N $ in the worst cast scenario since $ m_{i,h} $ is either a single number or a tuple containing 2 numbers which are all bounded above by $ N $.

However, the variables $ \mathcal{N}_i $, $ \mathcal{E}_i $ and $ \mathcal{R}_i $ are sets whose memory requirements are much larger than the previously mentioned single valued variables. The number of elements in $ \mathcal{N}_i $ is the degree of the node $ d_i $. Since each of the elements is the index of a node, the memory required to store each element is proportional to $ \log N $. Therefore, the memory requirement for the set is, $ M(\mathcal{N}_i) \sim 2 d_i \log N $. The maximum memory requirement for $ \mathcal{E}_i $ is exactly two times that of $ \mathcal{N}_i $ because the initial number of elements in $ \mathcal{E}_i $ and each element is a tuple of two numbers, each of which requires memory proportional to $ \log N $. Therefore, $ M(\mathcal{E}_i) \sim d_i \log N $. Note that, since $ \mathcal{N}_i $ is used only to construct the elements of $ \mathcal{E}_i $, the memory used for storing $ \mathcal{N}_i $ can simply be expanded to store $ \mathcal{E}_i $.  The elements of $ \mathcal{R}_i $ are also tuples whose memory requirements are similar to those in $ \mathcal{E}_i $. However, the number of such tuples in $ \mathcal{R}_i $ is equal to the number of residues $ r $ in the network, Hence, $ M(\mathcal{R}_i) \sim r \log N $. Therefore, the maximum memory requirement for each node $ i $ scales as,
\begin{equation}\label{eq:memory_costs}
M_i \sim \left( 2 d_i + r + 5 \right) \log N
\end{equation}
for large $ N $. In comparison, if the number of nodes is computed using the All-2-All broadcast method, each node requires memory that scales as $ N \log N $ as it needs to keep track of indices of every other node of the network. In the worst case scenario, the degree of each node of the network can be $ d_i = N-1 $ for all nodes which also implies that all nodes become residue nodes, yielding $ r = N $. In such a case, the memory requirement for the proposed algorithm scales as $ 3 N \log N $ for large $ N $ which is clearly worse than the All-2-All broadcast method. However, our numerical simulations involving much general classes of random networks show that $ r $ is at least one order of magnitude smaller than $ N $ (see Figure ~\Figure{} bottom panel) and the degree of each node is approximately the same (due to the parameters chosen in Table~\ref{tab:network_types}) for sufficiently large networks. Therefore, we can say that the proposed algorithm is better than the All-2-All broadcast method for sufficiently large networks.

\subsubsection{Memory Cost of the ST Algorithm}

For each query received, each node has to keep track of the id of the
querying node and of the id of its parent for the corresponding query.
Because there are $N$ queries and storing an id requires at least
$\log (N)$ bits, the memory required by each node to keep track of its
parents is at least $2N \log (N)$. Additionally, each node has to
ensure that it receives messages from all its neighbors for each
query. This requires an additional $d_i N$ bits, where $d_i$ is the
degree the node $i$. Therefore, the total memory required by the
generic node $i$ is $ 2N \log (N) + d_i N $.

\begin{table}
 \centering
 \caption{Asymptotic worst-case complexity for the AnB, the All-2-All,
   and the Single Tree (ST) algorithms in terms of time,
   communication, and memory. In the worst-case the three algorithms
   are comparable in every aspect (except for ST's memory) however,
   the precise memory and computational costs equations of
   Table~\ref{tab:generalcase} in the main text show that the AnB
   algorithm is more efficient in most cases.}\vspace{0.3cm}
\begin{tabular}{ |l|c|c|c|  }
 \hline
 \rowcolor{lightgray}
  Algorithm & Time & Communication & Memory \\
 \hline
 AnB& $\mathcal{O}(N)$ & $\mathcal{O}(N^2)$ & $\mathcal{O}(N \log(N))$\\
 All-2-All& $\mathcal{O}(N)$ & $\mathcal{O}(N^2)$ & $\mathcal{O}(N \log(N))$\\
 ST& $\mathcal{O}(N)$ & $\mathcal{O}(N^2)$  & $\mathcal{O}(N^2)$ \\
 \hline
\end{tabular}
\label{tab:worstcase}
\end{table}


\end{document}